\documentclass[11pt,a4paper,oneside]{article}
\usepackage[top=3cm, bottom=3cm, left=2cm, right=2cm]{geometry}
\usepackage[english]{babel}
\usepackage[utf8x]{inputenc}
\usepackage[T1]{fontenc}

\usepackage{amsthm,amsmath,amssymb,mathrsfs,dsfont,mathtools}

\usepackage{bm}
\usepackage{graphicx}
\usepackage[colorinlistoftodos]{todonotes}
\usepackage[colorlinks=true, allcolors=blue]{hyperref}
\usepackage{booktabs,subcaption}
\usepackage{hyperref}
\usepackage[all]{nowidow}
\usepackage{setspace}
\usepackage{algorithm}
\usepackage[noend]{algpseudocode}
\usepackage{tikz} 
\usepackage[sf]{titlesec}
\usepackage{sectsty}
\allsectionsfont{\centering \normalfont\scshape} 
\usepackage{lipsum}
\usepackage{adjustbox}
\usepackage{lineno} 
\usepackage{enumitem} 
\usepackage{natbib}

\usepackage{authblk}
\title{Factor pre-training in Bayesian multivariate logistic models}
\date{}
\author{Lorenzo Mauri}
\author{David B. Dunson}

\affil{Department of Statistical Science, Duke University, Durham, NC, 27708, U.S.A.}

\newtheorem{theorem}{Theorem}
\newtheorem{corollary}{Corollary}
\newtheorem{lemma}{Lemma}
\newtheorem{proposition}{Proposition}
\newtheorem{assumption}{Assumption}
\theoremstyle{definition}

\newtheorem{remark}{Remark}

\usepackage[resetlabels]{multibib}
\newcites{Supp}{References}

\usepackage{amsmath}

\usepackage{newtxtext}
\usepackage[subscriptcorrection]{newtxmath}

\usepackage{natbib}

\graphicspath{{./fig/}}

\usepackage{color}

\usepackage{hyperref}

\usepackage{graphicx}
\graphicspath{{fig/}}
\usepackage{caption}
\usepackage{subcaption}

\usepackage{dirtytalk}
\usepackage{hyperref}

\usepackage{algorithm}
\usepackage{algpseudocode}

\begin{document}

\maketitle

\begin{abstract}
This article focuses on inference in logistic regression for high-dimensional binary outcomes. A popular approach induces dependence across the outcomes by including latent factors in the linear predictor. Bayesian approaches are useful for characterizing uncertainty in inferring the regression coefficients, factors and loadings, while also incorporating hierarchical and shrinkage structure. However, Markov chain Monte Carlo algorithms for posterior computation face challenges in scaling to high-dimensional outcomes. Motivated by applications in ecology, we exploit a blessing of dimensionality to motivate pre-estimation of the latent factors. Conditionally on the factors, the outcomes are modeled via independent logistic regressions. We implement Gaussian approximations in parallel in inferring the posterior on the regression coefficients and loadings, including a simple adjustment to obtain credible intervals with valid frequentist coverage. We show posterior concentration properties and excellent empirical performance in simulations. The methods are applied to insect biodiversity data in Madagascar.
\end{abstract}

Ecology; 
Factor analysis; 
High-dimensional; 
Joint species distribution model; 
Latent variable model; 
Multivariate logistic regression

\section{Introduction}

High-dimensional binary observations are frequently recorded in ecological studies where the presence or absence of a large number of species $p$ is documented at $n$ sampling sites \citep{abrego_18, gssp, airborne_dna_fungi, coral}. Data consist of a $n \times p$ matrix $Y=[y_{ij}]_{ij}$, where $y_{ij}=1$ if the $j$-th species was observed in the $i$-th sample or $0$  otherwise. Ecologists are often interested in inferring relationships between sample-specific covariates $x_i=(x_{i1},\ldots,x_{iq})^\top$, such as temperature and precipitation, and species occurrences, while also inferring across-species dependence in occurrence. Such inferences are based on joint species distribution models, corresponding to multivariate binary regression models for $y_i=(y_{i1},\ldots,y_{ip})^\top$ given 
$x_i$ \citep{jsdm_warton, hmsc_time_series, hmsc, jsdm, jsdm_gpu, trace}. 

A popular approach induces dependence in species occurrence through a type of generalized linear latent variable model \citep{gllvm}, which lets
\begin{eqnarray}
\mbox{pr}(y_{ij}=1|x_i,\eta_i) = 
h(x_i^\top \beta_j + \lambda_j^\top \eta_i),\quad 
\eta_i \sim N_k(0, I_k),\quad 
(i=1,\ldots,n),\label{eq:gllvm}
\end{eqnarray}
where $h:\mathbb{R} \to (0,1)$ is a link function,
$B = \begin{bmatrix}
    \beta_{1} & \cdots & \beta_{p}
\end{bmatrix}^\top\in \mathbb R^{p \times q}$ is
a matrix of regression coefficients,
 $\eta_i \in \mathbb R^k$ are sample-specific latent factors, $\Lambda = \begin{bmatrix}
    \lambda_{1} & \cdots & \lambda_{p}
\end{bmatrix}^\top \in \mathbb R^{p \times k}$ are factor loadings and $p \gg k$. The factor term induces across-outcome dependence in a parsimonious manner. 

Calculating the likelihood of $y_i|x_i$ with $\eta_i$ marginalized out typically involves an intractable integral, motivating a literature on likelihood approximations. 
\citet{niku17,niku19} developed an efficient implementation of the Laplace approximation proposed by \citet{huber_laplace}, while
\citet{gllvm_va, gllvm_R, gllvm_eva} developed variational approximations. 
These methods take a few hours for each model fit for $p \approx 1,000$, leading to computational problems in our motivating applications, which have  $p=10,000-100,000$.
\citet{gmf} (\texttt{GMF}, henceforth) proposed a penalized quasi-likelihood estimator with better scalability, but their method takes up to several hours for each model fit, requires multiple fits for hyperparameter tuning, and does not provide uncertainty quantification without adding substantially to computation - for example, the authors suggest using bootstrap. The experiments in Section \ref{sec:simulation} and Section \ref{sec:additional_experiments} of the Supplementary Materials provide more details on running times. 

Bayesian inference can bypass calculating such integrals relying on data augmentation Gibbs samplers \citep{polya_gamma, albert_chib}, alternating between sampling the loading matrix, $\Lambda$, the latent factors, $\eta_i$'s, and other parameters including $B$ \citep{hmsc}. However, these algorithms suffer from poor mixing and slow convergence in high dimensions. Although pseudo-marginal algorithms can potentially be used relying on a Monte Carlo approximation to the marginal likelihood \citep{andrieu_roberts_09, doucet_et_al_15}, such algorithms have not been sufficiently scalable in our experiments. 
Alternatively, one can rely on analytic approximations to the posterior \citep{vi_review, ad_vi, bb_vi, stoch_vi, single_cell_fm, emp_bmf, vi_mfa}. For instance, variational inference
approximates the posterior with a more tractable distribution, but typically with little theoretical guarantees and severe under estimation of uncertainty.

An issue rendering frequentist and Bayesian methods impractical in high dimensions is the need to integrate out the latent factors. An alternative is to estimate latent factors, loadings, and coefficients jointly. Joint maximum likelihood estimates \citep{birnbaum_68} treat latent factors as fixed unknown parameters and can be computationally efficient using alternating optimization. Such approaches produce inconsistent estimates in the classic asymptotic regime where $n$ diverges and $p$ remains fixed \citep{haberman_77}. 
However, many applications are characterized by high-dimensional data with $p \gg n$ making the big $n$ fixed $p$ argument less relevant. \citet{chen_jmle, chen_identfiability} show consistency of joint maximum likelihood estimates as both $n$ and $p$ diverge, but without considering general covariate matrices or providing uncertainty quantification. 
Alternatively, \citet{fable, blast} propose fast approaches for Bayesian inference in linear single- and multi-study factor models for Gaussian data in which latent factors are estimated by singular value decomposition and loadings and residual variances are given a conjugate prior given the factors. They proved that the induced posterior on the covariance concentrates at the true values and has entry-wise credible intervals with correct coverage.  

Motivated by the above literature, 
we propose Fast multivariate Logistic Analysis for Inference in Regression (\texttt{FLAIR}). We first compute joint \textit{maximum a posteriori} estimates for latent factors, factor loadings, and regression coefficients. The posterior for the loadings and regression coefficients given the factors has a simple product form across the outcomes, and we approximate each term in this product via a Gaussian distribution. 
By including a careful variance inflation, we obtain credible intervals with valid frequentist 
for the regression coefficients and induced covariance of the linear predictor. Fixing latent variables at a point estimate in approximating the posterior of $\Lambda$ and $B$ can be justified because as $p$ grows, an increasing number of variables load on the latent factors and their marginal posterior concentrates. The product form allows \texttt{FLAIR} to be parallelized, making the implementation highly efficient on multicore machines. Hyperparameters are selected in a data-driven and automated manner. 

An anonymous referee pointed out the parallel development by \citet{lee24} (\texttt{LVHML}, henceforth). \texttt{LVHML} consider a longitudinal version of model \eqref{eq:gllvm} where the intercept term can vary over time. Hence, for the $j$-th element in the $i$-th sample and the $t$-th time,  \texttt{LVHML} lets 
\begin{equation*}
    pr(y_{ijt} = 1 \mid x_i) = h(\alpha_{jt} + x_i^\top \beta_j + \eta_i^\top \lambda_j), \quad (i=1, \dots, n; j=1, \dots, p; t=1, \dots, T), 
\end{equation*}
where $T$ is the total number of time points. \texttt{LVHML} assumes latent factors ($\{\eta_i\}_{i=1}^n$) are fixed unknown constants and imposes orthogonality between latent factors and covariates for identifiability. Fixed factors prevent the interpretation of $\Lambda \Lambda^\top$ as latent covariance between outcomes. \citet{lee24} provide central limit theorems (Theorems 3 and 6 in \citet{lee24}) to quantify uncertainty about point estimates. However, in our experiments, the confidence intervals obtained using these results suffered from non-negligible undercoverage, which seems to be persistent at different values of $n$ and $p$. We refer to Section \ref{subsec:longitudinal_experiments} in the Supplementary Material for additional details on the performance of \texttt{FLAIR} and \texttt{LVHML} in longitudinal scenarios.

Our contributions include: (i) an alternating optimization scheme to compute a point estimate, which is substantially faster than current alternatives for very large $p$ scenarios, with comparable or better accuracy, while requiring minimal or no hyperparameter tuning, (ii) theoretical support for our methodology showing consistency of joint {\em maximum a posteriori} and posterior point estimates as $n$ and $p$ diverge and posterior contraction around the truth, and (iii) a method to accurately quantify uncertainty without computationally expensive Markov chain Monte Carlo. Although we are motivated by ecology applications, our method is useful in a wide range of settings, from genetics \citep{stegle} to psychology \citep{skrondal}.

\section{Methodology}\label{sec:method}
\subsection{Notation}\label{subsec:notation}
We start by establishing the notation used in the paper. For a matrix $A$, we denote by $||A||_2$, $||A||_F$, $||A||_*$, $||A||_{\infty}$  its spectral, Frobenius, nuclear and entry-wise infinity norm, respectively, and by $s_l(A)$ its $l$-th largest singular value. For a vector $v$, we denote by $||v||$, $||v||_\infty$ its Euclidean and entry-wise infinity norm, respectively. 
Moreover, for two sequences $(a_n)_{n \leq 1}$, $(b_n)_{n \leq 1}$, we say $a_n \lesssim b_n$ if there exist two constants $N_0 < \infty$ and $C < \infty$, such that $a_n \leq C b_n$ for every $n > N_0$. We say $a_n \asymp b_n$ if and only if $a_n \lesssim b_n$ and $b_n \lesssim a_n$.

\subsection{General Approach}
We consider data generated from model 
\eqref{eq:gllvm} with 
$h^{-1}(\pi)=\log\{\pi/(1-\pi)\}$ the logit link function. Letting $z_{ij}=x_i^\top\beta_j + \eta_i^\top \lambda_j$ denote the linear predictor, and marginalizing out $\eta_i$, we get 
\begin{equation}
\label{eq:model_methodology_integrated}
    \begin{aligned}
        y_{ij} \mid z_{ij} & \sim \text{\text{Ber}}\big\{h(z_{ij})\big\}, \quad z_{i}\sim N_{p}\big(B x_i, \Lambda \Lambda^\top \big),
    \end{aligned}
    \quad (j=1, \dots, p; i=1, \dots, n),
\end{equation}
where $z_{i} = \left(z_{i1}, \dots, z_{ip}\right)^\top \in \mathbb R^p.$
The linear predictor for sample $i$, $z_i$, follows a $p$ dimensional singular Gaussian distribution with a rank $k$ covariance, $\Lambda \Lambda^\top$, which models across column dependence. Thus, marginal and co-occurrence probabilities depend uniquely on $B$ and $\Lambda \Lambda^\top$, with $\Lambda \Lambda^\top$ characterizing across outcome dependence not captured by covariate effects.
If the latent factors $M =  \begin{bmatrix}  \eta_1 & \cdots & \eta_n\end{bmatrix}^\top \in \mathbb R^{n \times k}$ were known, inference on the rows of $\Lambda$ and $B$ ($\{\lambda_j, \beta_j\}_{j=1}^p$) could be carried out by $p$ independent logistic regressions using the augmented covariate matrix $[X ~ M]$, since elements of $y_{i}$ are independent conditionally on the latent factor $\eta_i$.

Motivated by this consideration, we develop a computationally efficient approach to approximate the posterior of $B,\Lambda$. We first obtain a joint maximum {\em a posteriori} estimate for the latent factors, loadings, and regression coefficients using a combination of matrix factorization and optimization techniques. Then, given the estimated latent factors $\tilde M$, we characterize the uncertainty of $\Lambda$ and $B$ by their conditional posterior distribution. This conditional posterior is equivalent to the product of posteriors for $\{\lambda_j,\beta_j\}$ over $j=1,\ldots,p$. Each of these component posteriors can be calculated in parallel and accurately approximated with Gaussian distributions.

For high-dimensional data with large $p$ and $p\gg k$, many variables tend to load on each factor, leading to posterior concentration for each $\eta_i$. This blessing of dimensionality reduces concern about under-estimation of uncertainty due to fixing latent factors at a point estimate. We introduce an analytic inflation factor to the variance of the posterior, which can be calculated without tuning, ensuring valid frequentist coverage on average across credible intervals in all the experiments we considered, with coverage very close to the nominal level in each individual experiment. The complete procedure is reported in Algorithm \ref{alg:FLAIR}.

\subsection{Joint Maximum a Posteriori Estimate}\label{subsec:pre-estimation}
We assume truncated normal priors on the $\lambda_j$'s and $\beta_j$'s, \begin{equation}\label{eq:priors_lambda_beta} \begin{aligned} \lambda_j \mid \tau_{\lambda_j} \sim  TN_k\big(0, \tau_{\lambda_j }^2 I_k, [-c_\Lambda, c_\Lambda]^k \big), \quad \beta_j \mid \tau_{\beta_j} \sim TN_q\big(0, \tau_{\beta_j}^2 I_q, [-c_B, c_B]^q \big), \end{aligned}\end{equation}
for $j=1, \dots, p$, where $TN_{m}(\mu, \Sigma, C)$ denotes a $m$-dimensional truncated normal distribution with mean $\mu$ and covariance $\Sigma$ supported on the set $C$.
In addition, we use a truncated normal distribution for the latent factors,
\begin{equation}\label{eq:prior_factors}
   \eta_i \sim TN_{k}\big(0, I_k, [-2 \sqrt{\log(kn)}, 2 \sqrt{\log(kn)}]^k\big), \quad (i=1, \dots, n).
\end{equation}
$c_\Lambda$ and $c_B$ are user specified parameters that control the infinity norm of loadings and regression coefficients, respectively. Constraining the infinity norms of the model parameters and latent factors is useful for obtaining theoretical support for our methodology.
In our experiments, $c_\Lambda$ and $c_B$ are set by default to 10 to obtain a weak constraint. 
Computing the constrained joint maximum {\em a posteriori} estimate for $(M, \Lambda, B)$ under the priors specified in \eqref{eq:priors_lambda_beta}--\eqref{eq:prior_factors} is equivalent to solving the following constrained optimization problem,
\begin{equation}
\begin{aligned}
     (\hat M, \hat \Lambda, \hat B)  = &\arg\max_{M, \Lambda, B} \log p(M,\Lambda, B \mid  Y, X)\\
     &\text{s.t. } ||M||_{\infty} \leq 2 \log^{1/2}(kn), ||\Lambda||_{\infty} \leq c_\Lambda, ||B||_{\infty} \leq c_B,
\end{aligned}
  \label{eq:cjmap}
\end{equation}
where 
\begin{equation}\label{eq:log_joint_posterior}
\begin{aligned}
     \log p(M,\Lambda, B \mid  Y, X) =  & C + \sum_{i=1}^n \sum_{j=1}^p \log p(y_{ij} \mid x_i,  \lambda_j, \beta_j, \eta_i) \\
     &- \frac{1}{2}||M||_2^2 - \frac{1}{2}tr(\Lambda^\top \Sigma_{\Lambda}^{-1} \Lambda) - \frac{1}{2}tr(B^\top \Sigma_{B}^{-1} B),
\end{aligned}
\end{equation}
with 
$\Sigma_{\Lambda} = \text{diag}(\tau_{\lambda_1}^2, \dots, \tau_{\lambda_p}^2)$, $\Sigma_{B} = \text{diag}(\tau_{\beta_1}^2, \dots, \tau_{\beta_p}^2)$, and $C$ is a constant.
We solve \eqref{eq:cjmap} by iterating between the following steps until convergence.
\begin{enumerate}
    \item  Given the estimate for the latent factors $\hat M$, we update $B$ and $\Lambda$ via
    \begin{equation}\label{eq:optimization_params}
        (\hat \Lambda, \hat B)= \arg\max_{\Lambda, B} \log p(\hat M,\Lambda, B \mid  Y, X) \quad \text{s.t. } ||\Lambda||_{\infty} \leq c_\Lambda, ||B||_{\infty} \leq c_B,
    \end{equation}
     \item Given the estimate for the loadings and regression matrices $(\hat \Lambda, \hat B)$, we update $M$ via
    \begin{equation}\label{eq:optimization_factors}
        \hat M  = \arg\max_{M}\log p( M, \hat \Lambda, \hat B \mid  Y, X)  \quad \text{s.t. } ||M||_{\infty} \leq 2 \log^{1/2}(kn).
    \end{equation}
\end{enumerate}
Each step in the optimization algorithm is parallelizable across columns or rows of $Y$, is solved via a projected Newton-Raphson method, and has a cost of $\mathcal O\{np(k+q)^3 \texttt{max\_iter}\}$, where $\texttt{max\_iter}$ is an upper bound on the number of Newton steps of each routine. We stop iterations once the relative increase in the log-posterior is smaller than a small threshold $\epsilon$; in our experiments, we set $\epsilon=0.001$. 
The starting point for the algorithm is found via an initialization based on singular value decomposition adapted from \citet{chen_jmle}. This initialization was shown to provide consistent estimates for the loadings in \citet{ifa_svd}, when $X = 1_n$. 
More details are provided in the Supplementary Material. 
The solution to \eqref{eq:cjmap}$, (\hat M, \hat \Lambda, \hat B)$, is post-processed and transformed into the triplet $(\tilde M, \tilde \Lambda, \tilde B)$, so that $X \hat B ^\top + \hat M\hat \Lambda^\top = X\tilde B^\top+ \tilde M\tilde \Lambda^\top $, $\tilde M^\top \tilde M = nI_k$, and $\tilde M^\top X = 0$. This procedure is detailed in the Supplementary Material and leaves the value of the linear predictor unchanged, while enforcing the matrix product of the transpose of the latent factors with itself and with the matrix of covariates to be equal to their expectation. We take $\tilde M$ as our final estimate for $M$. 
We show that $\tilde \Lambda \tilde \Lambda^\top$ and $\tilde B$ are consistent in terms of the relative Frobenius error in Theorem \ref{thm:recovery_Z}. In the next section, we propose an approach for uncertainty quantification.

\subsection{Posterior Computation}\label{subsec:posterior}
In a Bayesian setting, uncertainty in $\Lambda$ and $B$ is encoded in their posterior distribution
\begin{equation}    \label{eq:true_posterior}
\begin{aligned}
    p(\Lambda, B \mid Y, X) 
    & \propto \int p(Y \mid X, \Lambda, B, M) p(B) p(\Lambda) p(M) d M \\
    & \propto \int p(\Lambda, B \mid Y, X, M)p(M \mid Y, X)dM. 
\end{aligned}
\end{equation}
As described in the introduction, approximating  \eqref{eq:true_posterior} via Markov chain Monte Carlo sampling is often impractical and existing alternatives are not satisfactory. We are motivated by the consideration that, when $p$ is large, estimates of latent factors become more accurate and their marginal posterior distribution, $p(M \mid Y, X)$, concentrates. Hence, we ignore uncertainty in $M$ and approximate the posterior distribution of $(\Lambda, B)$ via their conditional posterior distribution given the estimate for $M$, $\tilde M$,
\begin{equation}\label{eq:posterior_approx_1}
     p(\Lambda, B \mid Y, X) \approx p(\Lambda, B \mid Y, X, \tilde M).
\end{equation}

Conditionally on the latent factors, columns of $Y$ are independent. 
 Thus, for independent priors on rows of $\Lambda$ and $B$, the right hand side of \eqref{eq:posterior_approx_1} factorizes into the product of $p$ terms,
    \begin{equation}\label{eq:full_conditional}
     p(\Lambda, B \mid Y,  X, \tilde M) =  \prod_{j=1}^p p(\lambda_j, \beta_j \mid Y^{(j)},X, \tilde M),
    \end{equation}
    where $Y^{(j)}$ is the $j$-th column of $Y$, and 
    \begin{equation*}
       p(\lambda_j, \beta_j \mid Y^{(j)},X, \tilde M) \propto p(\lambda_j \mid \tau_{\lambda_j})p(\beta_j \mid \tau_{\beta_j}) \prod_{i=1}^n h \left(x_i^\top \beta_j + \tilde \eta_i^\top \lambda_j \right)^{y_{ij}}\left\{1-h \left(x_i^\top \beta_j + \tilde \eta_i^\top \lambda_j \right)\right\}^{1-y_{ij}}.
    \end{equation*} 
To speed-up posterior computation, we replace  conditional posterior distributions of $\theta_j =\left(\lambda_j, \beta_j\right)$'s given $M = \tilde M$, terms on the right hand side in \eqref{eq:full_conditional}, with Gaussian approximations,
\begin{equation*}\label{eq:pseudo_posterior_theta_j}
    \tilde \Pi_j( \theta_j) = N_{k+q}(\theta_j \mid  \tilde \theta_j, \rho^2 \tilde V_j), \quad (j=1, \dots, p),
\end{equation*} 
where $\tilde \theta_j = (\tilde \lambda_j^\top, \tilde \beta_j^\top)^\top$ is the estimate described in the previous section, 
\begin{equation}
\label{eq:V_j_tilde}
\tilde V_j = \bigg[-\frac{\partial^2}{\partial \theta_j \partial \theta_j^\top}\log p( Y^{(j)} \mid X, \tilde M, \lambda_j, \beta_j) + \log p(\lambda_j \mid \hat \tau_{\lambda_j}) + \log p(\beta_j \mid \hat \tau_{\beta_j}) \mid_{\theta_j = \tilde \theta_j} \bigg]^{-1}
\end{equation}
is the corresponding inverse negative Hessian and $\rho>1$ is a variance inflation factor that is fixed in advance to achieve correct frequentist coverage of posterior credible intervals. Section \ref{subsec:rho} of the Supplementary Materials presents a closed-form formula for $\rho$. In our experiments, the credible intervals obtained from $\tilde \Pi$ for $\Lambda \Lambda^\top$ and $B$ had accurate frequentist coverage.

We approximate the posterior for $(\Lambda,B)$ in \eqref{eq:true_posterior} by approximating
\eqref{eq:full_conditional} via
\begin{equation}\label{eq:pi_tilde}
    \tilde \Pi(\theta_1, \dots, \theta_p) = \prod_{j=1}^p \tilde \Pi_j(\theta_j).
\end{equation}
Hyperparameters $\tau_B =(\tau_{\beta_1}, \dots, \tau_{\beta_p})$ and $\tau_\Lambda =(\tau_{\lambda_1}, \dots, \tau_{\lambda_p})$ are selected using a data-driven strategy reported in the Supplementary Material.

\subsection{Choice of the Number of Latent Factors}\label{subsec:choice_k}
To select the number of factors, we opt for the joint likelihood-based information criterion introduced in \cite{chen_jic} that takes the form
\begin{equation*}\label{eq:JIC}
    \text{JIC}(k) = -2 l_k + k \max(n,p) \log \{\min(n,p)\},
\end{equation*}
where $l_k$ is the value of the joint log-likelihood computed at the joint maximum likelihood estimate when the latent dimension is equal to $k$. To avoid computing the joint maximum likelihood estimate for every value of $k$, we approximate $ l_k$ with $ l_k \approx \hat l_k = \log p(Y \mid X, \hat M_k, \hat \Lambda_k, \hat B_k)$, where $( \hat M_k, \hat \Lambda_k, \hat B_k)$ are obtained via the matrix-factorization technique described in Section \ref{subsec:svd_init} of the supplemental. Thus, we set
\begin{equation}\label{eq:k_hat}
    \hat k = \arg \min_{k=1, \dots, k_{max}} \hat{ \text{JIC}}(k), \quad \hat{ \text{JIC}}(k)= -2 \hat l_k + k \max(n,p) \log\{\min(n,p)\},
\end{equation}
where $k_{max}$ is an upperbound to the latent dimension. The criterion in \eqref{eq:k_hat} always picked the correct number of latent factors in the simulations reported in Section \ref{sec:simulation}.

\begin{algorithm}
\caption{\texttt{FLAIR} procedure to obtain $N_{MC}$ approximate posterior samples.}\label{alg:FLAIR}
\begin{algorithmic}
   \Require The data matrix $Y \in \mathbb{R}^{n \times p}$, the design matrix $X$, the number of Monte Carlo samples $N_{MC}$, the upper bound on the number of factors $k_{max}$, and the upper bound on the magnitude of the regression coefficients and factor loadings $c_B, c_\Lambda$.\State \State \textbf{Step 1:} Estimate the number of latent factors via equation \eqref{eq:k_hat}. 

   \State \textbf{Step 2:} Obtain initial estimates $\hat B,\hat \Lambda$ as described in Section \ref{subsec:svd_init} in the supplemental and let $\hat \tau_{\lambda_j} = \mathcal{T}\left( k^{-1/2}\left|\left|\hat \lambda_j\right|\right|\right)$ and $\hat \tau_{\beta_j} = \mathcal{T}\left(  k^{-1/2}\left|\left|\hat \beta_j\right|\right|\right)$, where $\mathcal{T}$ is defined as $\mathcal{T}(x) = x \mathtt{1}_{\{x \in (l,u) \}} + l\mathtt{1}_{\{x \leq l \}} + u \mathtt{1}_{\{x \geq u \}}$ and $l = 0.5, u=20$.  
\State \textbf{Step 3:} Compute 
    $ (\hat M, \hat \Lambda, \hat B)$ by solving 
    \eqref{eq:cjmap}. 
    \State \textbf{Step 4:} Post-process $ (\hat M, \hat \Lambda, \hat B)$ to obtain $ \left(\tilde M, \tilde \Lambda, \tilde B\right)$  as described in Section \ref{subsec:postprocessing} of the Supplementary Material.
    \State \textbf{Step 5:} Estimate the coverage-correction factor $\rho$ as described in Section \ref{subsec:rho} of the Supplementary Material.
\State \textbf{Step 6:}\For{$j=1$ to $p$ \textbf{in parallel}}
    \For{$s=1$ to $N_{MC}$}
    \State Sample independently $\theta_j^{(s)} = (\beta_j^{(s)}, \lambda_j^{(s)})$ from
$\theta_j^{(s)} \sim N_{k+q}\left( \tilde \theta_j, \rho^2 \tilde V_j\right)$, $\tilde \theta_j = (\tilde \lambda_j^\top, \tilde \beta_j^\top)^\top$, $\tilde V_j$ is defined \\\hspace{1.1cm} in \eqref{eq:V_j_tilde}.
\State Form $B^{(s)} = \big\{
    \beta_1^{(s)}, \cdots, \beta_p^{(s)}
\big\}^\top$ and  
$\Lambda^{(s)} = \big\{
    \lambda_1^{(s)}, \cdots, \lambda_p^{(s)}
\big\}^\top$.
\State Compute the corresponding sample for the latent covariance matrix as $\Lambda^{(s)}\Lambda^{(s)\top}$.
 \EndFor
\EndFor
\State \textbf{Output:} $N_{MC}$ samples of the covariance matrix $\Lambda^{(1)}\Lambda^{(1)\top}, \ldots, \Lambda^{(N_{MC})}\Lambda^{(N_{MC})\top}$ and of
 the regression coefficients matrix $B^{(1)}, \ldots,B^{(N_{MC})} $.

\end{algorithmic}
\end{algorithm}

\section{Theoretical Support}\label{sec:theory}
 Before stating the theoretical results, we enumerate some regularity conditions.
\begin{assumption}\label{assumption:dgm}
 The data are generated according 
 \eqref{eq:gllvm} with 

true parameters $B=B_0 = 
    (\beta_{01} \cdots \beta_{0p})^\top\in \mathbb R^{p \times q}$ and $\Lambda=\Lambda_0 = (
    \lambda_{01} \cdots  \lambda_{0p})^\top \in \mathbb R^{p \times k}$. We denote by $M_0$ the matrix whose rows are given by the true latent factors $\eta_{0i}$'s and define $Z_0 = M_0 \Lambda^\top + X B_0^\top$.
\end{assumption}

\begin{assumption}\label{assumption:p}Data dimensionality $p = p_n$ satisfies $p_n \to \infty$, $\log(p_n)/n = o(1)$, $p_n \gtrsim n^{1/2}$.
\end{assumption}
Assumption \ref{assumption:p} ensures that the number of outcomes $p$ grows asymptotically with $n$ at most at a polynomial rate and not slower than its square root.
\begin{assumption}\label{assumption:lambda_sv}
   The true loadings $\Lambda_0$ satisfy $s_k(\Lambda_0) \asymp p_n^{1/2}$ as $n\to \infty$ and $||\Lambda_0||_{\infty} \leq c_\Lambda < \infty$.
\end{assumption}

\begin{assumption} \label{assumption:B}
The true regression coefficients matrix $B_0$ satisfies $||B_0||_{\infty} \leq c_B$.
\end{assumption}
\begin{assumption}\label{assumption:X}
The design matrix $X$ satisfies $||X||_{\infty} \leq c_x \log^{1/2}(qn)$ with probability at least $1 - C'/n$ and $s_q \left(X^\top X\right) \asymp n$, where $C'$ is a positive constant not depending on $n$ and $p_n$.
\end{assumption}
Assumption \ref{assumption:X} holds if, for instance, the $x_i$'s are independent samples from a sub-Gaussian distribution.
Assumptions \ref{assumption:lambda_sv} -- \ref{assumption:X} ensure $||Z_0||_{\infty} \lesssim \log^{1/2}\{(k \vee q) n\}$, with high probability, since, as we show in the supplement $||M_0||_{\infty} \leq 2\log^{1/2}(kn)$, with probability at least $1- 2/n$. 

\begin{assumption}\label{assumption:hyperparametes}
    The hyperparameters $\tau_B, \tau_\Lambda, k, \rho$ are fixed constants. 
\end{assumption}

First, we show that the accuracy of point estimates improves as $n,p$ diverge. Treating the distribution arising from \eqref{eq:pi_tilde} as our  posterior distribution, the posterior mean of $\Lambda \Lambda^\top$ and $B$ are available in closed form once the $\tilde \lambda_j$'s, $\tilde V_j$'s, and $\tilde \beta_j$'s have been computed and are given by
\begin{equation}\label{eq:expectation_tilde}
   \tilde \Sigma =  E_{\tilde \Pi}(\Lambda \Lambda^\top) = \tilde \Lambda \tilde \Lambda^\top + \rho^2 \tilde D,  \quad   E_{\tilde \Pi}(B) = \tilde B,
\end{equation}
where $E_{\tilde \Pi}$ denotes the expectation under $\tilde \Pi$, $\tilde D = \text{diag}\big\{ tr(\tilde V_{\lambda_1}), \dots, tr(\tilde V_{\lambda_p}) \big\}$,
$\tilde V_{\lambda_j}$ is the marginal posterior variance of $\lambda_j$ from $\tilde V_j$, and $\tilde \Lambda$ and $\tilde B$ are matrices whose $j$-th rows are 
$\tilde \lambda_j$ and $\tilde \beta_j$ respectively. We show that $\tilde \Sigma$ and $\tilde B$ are consistent estimates; $\tilde{\Sigma}$ has a
low-rank plus diagonal form. An alternative low-rank estimator can be obtained by discarding the diagonal part and estimating
$\Lambda_0 \Lambda_0^\top$ via $\tilde \Lambda \tilde \Lambda^\top$; this estimator has comparable asymptotic performance as shown below.

\begin{theorem}[Accuracy of Point Estimates]
    \label{thm:recovery_Z}
    Suppose Assumptions \ref{assumption:dgm} -- \ref{assumption:hyperparametes} hold and define $\hat Z =X \hat B^\top + \hat M \hat \Lambda^\top$, where $(\hat M, \hat \Lambda, \hat B)$ is the solution to \eqref{eq:cjmap}. Then, with probability at least $1- C/n$,
    \begin{equation}\label{eq:parameters_cjmap}
      \frac{1}{(np)^{1/2}}\left|\left|\hat Z - Z_0\right|\right|_F \lesssim   e^{c_Z \log^{1/2}\{(k \vee q)n\}} \bigg\{\frac{1}{n^{1/2}} + \frac{\log^{1/2}(kn)}{p_n^{1/2}} \bigg\}.
    \end{equation}
    Moreover, define $(\tilde M, \tilde \Lambda, \tilde B)$ as the triplet obtained post-processing the joint maximum a posteriori estimate $(\hat M, \hat \Lambda, \hat B)$ solution to \eqref{eq:cjmap}. 
    Then, with probability at least $1- C/n$, we have
    \begin{align}
     \min_{R \in \mathbb R^{k \times k}: R^\top R = I_k} \frac{1}{n} \left|\left|\tilde M R - M_0\right|\right|_F 
     &\lesssim   e^{c_Z \log^{1/2}\{(k \vee q)n\}} \bigg\{\frac{1}{n^{1/2}} + \frac{\log^{1/2}(kn)}{p_n^{1/2}}\bigg\},\label{eq:recovery_P_U}\\
      \frac{\left|\left|\bar \Sigma - \Lambda_0 \Lambda_0^\top \right|\right|_F}{\left|\left| \Lambda_0 \Lambda_0^\top \right|\right|_F}&\lesssim e^{c_Z \log^{1/2}\{(k \vee q)n\}} \bigg\{\frac{1}{n^{1/2}} + \frac{\log^{1/2}(kn)}{p_n^{1/2}} \bigg\} \label{eq:posterior_mean_Lambda_outer_accuracy},\\
      \frac{\left|\left|\tilde \Lambda \tilde \Lambda^\top - \Lambda_0 \Lambda_0^\top \right|\right|_F}{\left|\left| \Lambda_0 \Lambda_0^\top \right|\right|_F}&\lesssim  e^{c_Z \log^{1/2}\{(k \vee q)n\}} \bigg\{\frac{1}{n^{1/2}} + \frac{\log^{1/2}(kn)}{p_n^{1/2}} \bigg\}\label{eq:recovery_Lambda_outer}, \\
       \frac{1}{(p_nq)^{1/2}}\left|\left|\tilde B - B_0\right|\right|_F&\lesssim  e^{c_Z \log^{1/2}\{(k \vee q)n\}} \bigg\{\frac{1}{n^{1/2}} + \frac{\log^{1/2}(kn)}{p_n^{1/2}} \bigg\}\label{eq:recovery_B},
    \end{align}
where $C$ and $c_Z$  are universal constants not depending on $n$ and $p_n$.
\end{theorem}

\begin{proof}
    The proofs of all theoretical results are reported in the Supplementary Material.
\end{proof}
\begin{remark}[Blessing of Dimensionality]
    The magnitude of the errors of the joint maximum {\em a posteriori} estimates decrease at a rate $\mathcal O\Big(\frac{1}{n^{1/2}} + \frac{1}{p_n^{1/2}}\Big)$ up to subpolynomal factors. Hence, we have a blessing of dimensionality with consistency holding only if both $n$ and $p_n$ diverge.
\end{remark}

\begin{remark}[Convergence Rate]
The bound in Theorem \ref{thm:recovery_Z} is less tight than the related bounds in \cite{chen_jmle, lee24}. This is due to different assumptions on latent factors and covariates. These works assume that $\eta_i$' s and $x_i$' s lie in a compact space that is not dependent on $n$ and $p$, while our assumptions are more general, leading to the factor $ e^{c_Z \log^{1/2}\{(k \vee q)n\}}$ on the right-hand side of Theorem \ref{thm:recovery_Z}. We refer to the note after the proof of Theorem \ref{thm:recovery_Z} in the Supplementary Material for a more detailed explanation.
\end{remark}

The right-hand side of \eqref{eq:parameters_cjmap} is better than the rate in \citet{chen_jmle} for the joint maximum likelihood estimate and the same up to subpolynomial terms as \citet{chen_identfiability}, 
which focuses on confirmatory factor analysis with $X = 1_n$. The result in \eqref{eq:recovery_P_U} bounds the Procrustes error of the estimate for the latent factors $\tilde M$. In \eqref{eq:posterior_mean_Lambda_outer_accuracy}--\eqref{eq:recovery_Lambda_outer} and \eqref{eq:recovery_B}, we normalize the left-hand side by dividing by the norm of $\Lambda_0 \Lambda_0^\top$ and $(p_n q)^{1/2}$ to make the estimation error comparable as the dimension $p_n$ increases.

Next, we characterize the contraction of the posterior distribution $\tilde \Pi$ around the true parameters. 
\begin{theorem}[Posterior Contraction]\label{thm:posterior_contraction}
  Suppose Assumptions \ref{assumption:dgm} -- \ref{assumption:hyperparametes} hold. Then, for $M \in \mathbb R$ sufficiently large, we have 
     \begin{align}
        \text{pr} \bigg[ \tilde \Pi\bigg\{\frac{\big|\big| \Lambda  \Lambda^\top - \Lambda_0 \Lambda_0^\top \big|\big|_F}{\big|\big| \Lambda_0 \Lambda_0^\top \big|\big|_F}> M\epsilon_n\bigg\} \leq C/n\bigg]&\geq 1 - C/n\\
        \text{pr} \bigg[ \tilde \Pi\bigg\{\frac{1}{\sqrt{p_n q}} \big|\big| B - B_0\big|\big|_F>  M\epsilon_n \bigg\} \leq C/n\bigg]&\geq 1 - C/n 
    \end{align}
    where  $\epsilon_n=e^{c_Z \log^{1/2}\{(k \vee q)n\}}\left\{\frac{\log^{1/2}(p_n)}{n^{1/2}} + \frac{\log^{1/2}(kn)}{p_n^{1/2}} \right\}$,
    $\text{pr}$ and $\tilde \Pi$ denote the true data generating probability measure and the posterior probability measure induced by \eqref{eq:pi_tilde} respectively, and $C$ and $c_Z$ are universal constants not depending on $n$ and $p_n$.
\end{theorem}


For both $\Lambda_0 \Lambda_0^\top$ and $B_0$, we rescale the distance by suitable quantities to take into account the growing dimension of the parameter space. The contraction rate is given by $\frac{1}{n^{1/2}} + \frac{1}{p_n^{1/2}}$ up to subpolynomal factors and is the same, modulo a logarithmic term, as rates for the estimates.

\section{Numerical Experiments} \label{sec:simulation}
We conduct a simulation study to illustrate the performance of \texttt{FLAIR} in estimation accuracy and uncertainty quantification for 
 $\Lambda_0 \Lambda_0^\top$ and $B_0$, as well as computing time.
We simulate data from model \eqref{eq:gllvm}, where parameters are generated as follows
\begin{equation*}\begin{aligned}
\lambda_{0jl} &\sim 0.5 \delta_0  + 0.5 TN(0, \sigma^2, [-5, 5]), \quad \beta_{0jl'} \sim  0.5 \delta_0  + 0.5TN(0,\sigma^2, [-5, 5])\\ 
\end{aligned}  \end{equation*} for $j=1, \dots, p$, $l=1, \dots, k$, $l'=1, \dots, q$. 
We let the sample and outcome sizes be $(n, p) \in \{500, 1000\} \times \{1000, 5000, 10000\}$, and set $\sigma^2 = 0.5$, $k=q=100$. 

For each configuration, we replicate the experiments 50 times. In each replicate, covariates and latent factors are generated as 
\begin{equation*}
    x_i = (1, x_{i2}, \dots x_{iq}), \quad  x_{ij}\sim N(0,1),
    \quad \eta_i\sim N_k(0,I_k) \quad (j=2, \dots, q; i=1, \dots, n)
\end{equation*}
We evaluate estimation accuracy for $\Lambda_0 \Lambda_0^\top$ and $B_0$ via the Frobenius norm of the difference of the estimate and true parameter scaled by $\left|\left| \Lambda_0 \Lambda_0^\top\right| \right|_F$ and $(pq)^{1/2}$ respectively, as in Section \ref{sec:theory}. We evaluated uncertainty quantification through the average frequentist coverage of equal-tail $95\%$ credible intervals for individual parameters. For \texttt{FLAIR}, we use posterior means as point estimates, but other possible estimates mentioned above had similar performance.

We compare to \texttt{GMF} using code  at \href{https://github.com/kidzik/gmf}{https://github.com/kidzik/gmf}, using either a Newton method with a simplified Hessian or alternating iteratively reweighted least squares. These two approaches had substantially different computing time and estimation accuracy, so we report results for both. For each replicate, we performed a random $80\%/20\%$ train-test split and chose the hyperparameters to maximize the test set area under the ROC curve; then we re-fitted the model with the full data. In the supplemental, we consider scenarios with lower-dimensional parameters, longitudinal data, and without covariates, where we also include a comparison with standard implementations of generalized linear latent variable models \citep{gllvm_va, gllvm_eva}, \texttt{LVHML}, and \citet{chen_jmle}'s method, respectively.
\begin{table}
\centering	
{		\begin{tabular}{crrrcrrr}
		&  \multicolumn{7}{c}{$p=1000$}  \\
& \multicolumn{3}{c}{$n=500$}& &\multicolumn{3}{c}{$n=1000$} \\
		Method & $\Lambda \Lambda^\top$ & $B$ & time (s)  & & $\Lambda \Lambda^\top$ & $B$ & time (s) \\
		\texttt{GMF - Newton} & $44.36^{1.15}$  & $14.49^{0.07}$ & $35.50^{0.69}$ & & $29.13^{0.88}$ &  $10.19^{0.09}$ & $76.16^{1.63}$\\ 
  		\texttt{GMF - Airwls} & $41.79^{0.18}$ & $14.01^{0.03}$ & $221.45^{11.83}$ &  &$>100$  & $>100$ & $495.30^{243.58}$  \\ 
		\texttt{FLAIR} & $38.95^{0.12}$  & $14.35^{0.04}$ & $5.41^{0.21}$ & & $27.29^{0.07}$ & $10.17^{0.03}$ &$12.05^{0.42}$ \\ 
 &  \multicolumn{7}{c}{$p=5000$}  \\
& \multicolumn{3}{c}{$n=500$}& &\multicolumn{3}{c}{$n=1000$} \\
 Method & $\Lambda \Lambda^\top$ & $B$ & time (s) & & $\Lambda \Lambda^\top$ & $B$ & time (s) \\
		\texttt{GMF - Newton} & $41.46^{1.00}$ & $14.27^{0.06}$& $122.17^{3.65}$ &  & $28.31^{0.88}$ & $9.95^{0.09}$& $140.78^{4.92}$  \\ 
  		\texttt{GMF - Airwls} & $41.20^{0.14}$ & $14.05^{0.03}$ & $1729.77^{149.25}$  &  & $>100$ 
    &  $>100$ & $3854.37^{299.71}$ \\ 
		\texttt{FLAIR} & $39.30^{0.12}$&  $14.26^{0.04}$ &$19.52^{0.73}$ &  & $27.33^{0.07}$  & $10.09^{0.03}$ & $25.25^{0.85}$\\ 
&  \multicolumn{7}{c}{$p=10000$}  \\
& \multicolumn{3}{c}{$n=500$}& &\multicolumn{3}{c}{$n=1000$} \\
 Method & $\Lambda \Lambda^\top$ & $B$ & time (s) & & $\Lambda \Lambda^\top$ & $B$ & time (s) \\
		\texttt{GMF - Newton} & $41.44^{0.56}$ & $14.24^{0.04}$& $175.00^{9.93}$ & & $29.97^{1.45}$ & $10.10^{0.13}$ & $450.86^{9.36}$  \\ 
  		\texttt{GMF - Airwls} &  $44.53^{0.45}$& $14.38^{0.05}$ & $7662.75^{283.39}$  & &$28.81^{0.48}$ & $10.22^{0.29}$ &$19418.37^{590.68}$ \\ 
		\texttt{FLAIR} & $39.46^{0.12}$ & $14.26^{0.04}$  & $35.72^{0.92}$& & $27.44^{0.07}$  & $10.09^{0.03}$ &$53.17^{1.11}$\\ 
	\end{tabular}}
	\label{tablelabel}
	\caption{Comparison of the methods in terms of estimation accuracy.    Root normalized squared error for $\Lambda \Lambda^\top$ and $B$, and running time. Estimation errors have been multiplied by $10^2$. We report mean and standard error over 50 replications. \texttt{GMF - Newton} and \texttt{GMF - Airwls} denote \citet{gmf}'s method fitted via the quasi Newton method and via alternating iteratively reweighted least square algorithm respectively.}
	\label{tab:accuracy}
\end{table}

Table \ref{tab:accuracy} reports a comparison in terms of estimation accuracy and computational time. \texttt{FLAIR} is remarkably faster than \texttt{GMF}, even with the results in the table focusing only on model fitting time after hyperparameter tuning. In terms of estimation accuracy, \texttt{FLAIR} 
has a better performance in estimating $\Lambda \Lambda^\top$ while being comparable in estimating $B$.
\texttt{GMF} fitted via iterated least squares had extremely poor accuracy in some replicates affecting the overall performance when $(n, p) \in \{1000\} \times\{1000, 5000\}$.
Table \ref{tab:uq} reports the coverage of credible intervals on average across the entries of $B$ and $\Lambda \Lambda^\top$. These results provide
strong support for \texttt{FLAIR} in terms of providing well-calibrated credible intervals. 

\begin{table}\centering	
{	\begin{tabular}{crrcrr}
&  \multicolumn{5}{c}{$p=1000$}  \\
& \multicolumn{2}{r}{$n=500$}& &\multicolumn{2}{r}{$n=1000$} \\
		Method & $\Lambda \Lambda^\top$ & $B$ & & $\Lambda \Lambda^\top$ & $B$ \\ 
		\texttt{FLAIR} & $96.70^{0.10}$ & $95.31^{0.09}$ &  & $96.42^{0.09}$ & $95.24^{0.09}$ \\ 
		vanilla \texttt{FLAIR} ($\rho=1$) & $92.57^{0.14}$ & $90.99^{0.10}$ & & $92.23^{0.13}$  & $90.08^{0.11}$\\ 
 &  \multicolumn{5}{c}{$p=5000$}  \\& 
 \multicolumn{2}{r}{$n=500$}& &\multicolumn{2}{r}{$n=1000$}\\
 		Method & $\Lambda \Lambda^\top$ & $B$ & & $\Lambda \Lambda^\top$ & $B$ \\ 
 \texttt{FLAIR} & $96.47^{0.10}$ &  $94.04^{0.10}$ &  & $96.16^{0.09}$ &$93.75^{0.09}$  \\ 
		vanilla \texttt{FLAIR} ($\rho=1$) & $92.85^{0.13}$  & $89.81^{0.10}$ & & $92.68^{0.12}$  & $89.66^{0.10}$ \\ 
   &  \multicolumn{5}{c}{$p=10000$}  \\
& \multicolumn{2}{r}{$n=500$}& &\multicolumn{2}{r}{$n=1000$}\\
		Method & $\Lambda \Lambda^\top$ & $B$ & & $\Lambda \Lambda^\top$ & $B$ \\ 
\texttt{FLAIR} & $96.17^{0.07}$ & $95.20^{0.08}$  &  & $95.89^{0.08}$& $95.02^{0.06}$\\ 
		vanilla \texttt{FLAIR} ($\rho=1$) & $92.71^{0.13}$ & $91.79^{0.10}$ & & $92.63^{0.12}$  & $91.70^{0.09}$\\ 
	\end{tabular}}
	\label{tablelabel}
\caption{Frequentist coverage of 95\% credible intervals for individual parameters by \texttt{FLAIR} with and without applying the correction factor $\rho$ to the posterior variance. Average frequentist coverage for entries of a random $100\times 100$ submatrix of $\Lambda \Lambda^\top$ and $B$ for equi-tailed 95\% credible intervals of \texttt{FLAIR} with and without applying the correction factor $\rho$ to the posterior variance. We report mean and standard error over 50 replications. All values have been multiplied by $10^2$.}
	\label{tab:uq}
\end{table}

\section{Application to Madagascar Arthropod Data}\label{sec:application}
We analyze data from \citet{coral} measuring arthropod co-occurence. Arthropods are a vital component of any ecosystem, and characterizing their co-occurrence is of paramount importance in studying factors related to community assembly and biodiversity. Data were collected from 284 samples at 53 sampling sites in Madagascar. At each sampling site, arthropods were collected in Malaise traps and categorized according to their DNA through COI metabarcoding \citep{coi_meta} and the OptimOTU pipeline \citep{dna_barcoding}. This produced 254312 operational taxonomic units, which we refer to as \say{species}. Most are ultra-rare, with 211187 of these species present in $\le 2$ samples. As covariates,
we included log-transformed sequencing depth, mean precipitation and temperature, their interaction and squares, and four trigonometric terms to adjust for seasonal effects ($\cos(2l\pi d_i/365), \sin(2l\pi d_i/365)$ with $l=1,2$, where $d_i$ denotes the day of sampling for the $i$-th observation). We standardized continuous covariates to have zero mean and unit standard deviation. 

To allow comparisons with less computationally efficient alternatives, we initially focused our analysis on the $5656$ species that were observed at least 15 times. Using the approach of Section \ref{subsec:choice_k}, the estimated number of latent factors was $\hat k = 7$. We applied a random stratified $80\%-20\%$ split to the data set, with stratification ensuring that the holdout set contains roughly the same proportions of 0 and 1s as the training set. To choose the hyperparameters of \texttt{GMF}, we divided the holdout set into half into test and validation sets. All details are in the Supplement. 

Obtaining \texttt{FLAIR} estimates took $\sim 15$ minutes while the average running time of \texttt{GMF} for each hyperparameter configuration was approximately 1 hour with the quasi-Newton algorithm and more than 5 hours with the iterated least squares algorithm. 
\texttt{FLAIR} obtained an area under the curve on the validation set of 96.53\%, while \texttt{GMF} with the best hyperparameter configuration yielded 95.30\% and 78.28\% for the Newton method and iterated least squares algorithm respectively. Hence, \texttt{FLAIR} achieved better predictive performance with considerable less computing time. 

We reanalyzed the data including the $43125$ species that were observed at least 3 times. For \texttt{GMF}, we did not optimize the hyperparameters again and fitted the model using the configurations chosen in the common species analysis described above and focused on the faster and more accurate quasi-Newton algorithm. \texttt{FLAIR} had considerably better out-of-sample predictive performance, having an area under the curve of 94.18\% compared to the 87.86\% for \texttt{GMF}. As expected, performance dropped off somewhat compared to the above common species analysis, since rare species are more difficult to predict.

\begin{figure}
\centering

\begin{subfigure}{1\textwidth}
    \centering
    \includegraphics[width=0.8\linewidth]{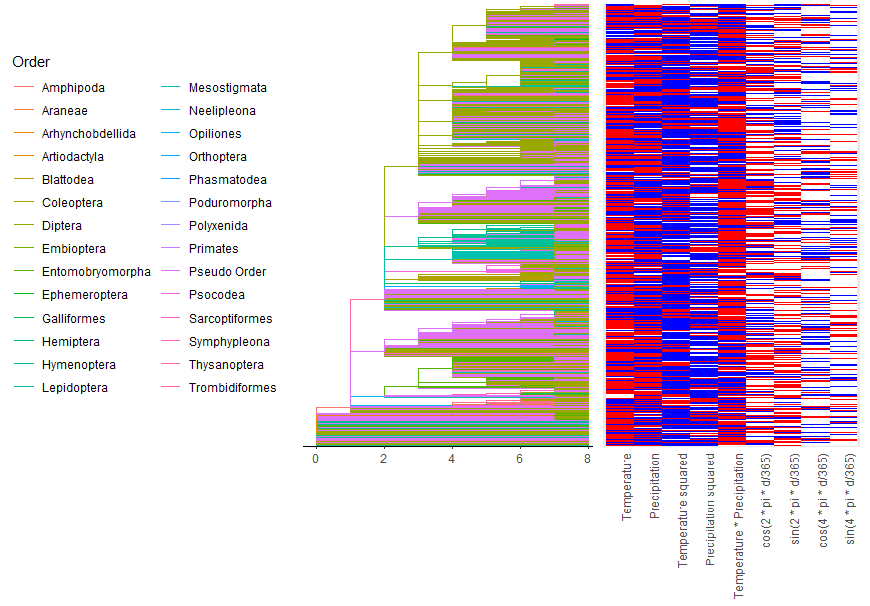}
    \label{fig:Beta_50_zeros}
\end{subfigure}
\vfill
\centering
\caption{Taxonomical tree of observed species (left panel) and responses of the species to measured covariates (right panel). Responses that were estimated to be positive (red) or negative (blue) with at least 95\% posterior probability.}
\label{fig:insects_Beta}
\end{figure}

In conducting inferences on the cross-species dependence in occurrence and covariate effects, we focused on an analysis of the complete data on the 5656 common species. Figure \ref{fig:insects_Beta} illustrates species responses to measured covariates. The results imply that most arthropod species are sensitive to climate, including both temperature and precipitation, and seasonality. In addition, the impact of temperature and precipitation tends to be nonlinear. However, there is substantial heterogeneity among species in the signs of the coefficients, suggesting that ideal climate conditions are species-specific. There is no clear taxonomic clustering in the signs, suggesting that even closely related species may have different ideal climate conditions. 

\begin{figure}
\centering
\begin{subfigure}{1\textwidth}
    \centering
    \includegraphics[width=0.95\linewidth]{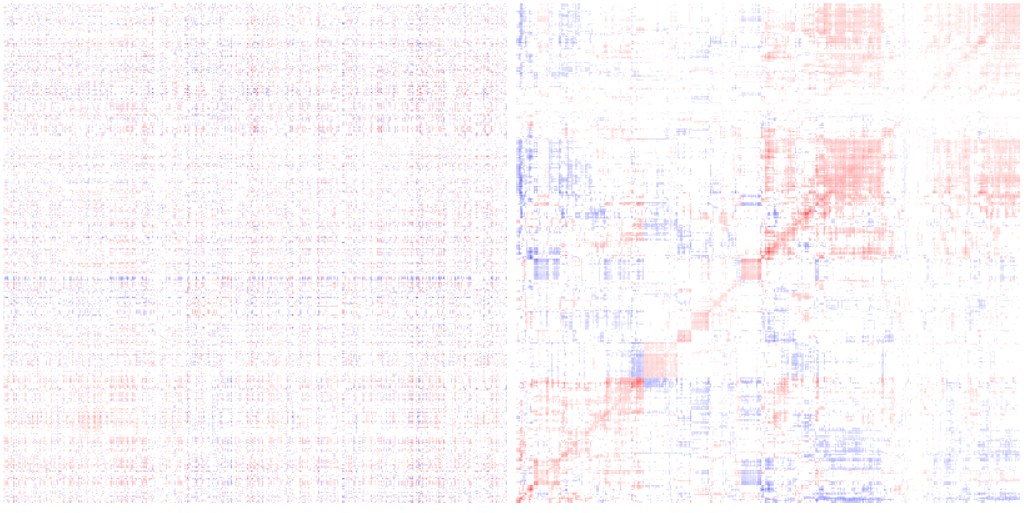}
\end{subfigure}
\vfill
\centering
\caption{Posterior mean of $\Lambda \Lambda^\top$ where entries for which the $95\%$ credible interval included $0$ were set to $0$ for 1000 species ordered according the taxonomical tree (left panel) and to their pairwise latent covariance (right panel). Red (blue) denotes positive (negative) values.}
\label{fig:insects_Lambda}
\end{figure}

We are also interested in cross-species dependence in co-occurrence, which is captured by the $\Lambda \Lambda^\top$ matrix. If the $j,j'$ entry of this matrix is positive, then that pair of species may prefer similar habitat conditions or may have beneficial interactions. If the entry is negative, the two species may have adversarial interactions in competing for the same resources or having a predator-prey relationship, or may favor different conditions. Figure \ref{fig:insects_Lambda} shows the posterior mean of $\Lambda \Lambda^\top$ for 1000
randomly selected species; entries for which the $95\%$ credible interval included $0$ were set to $0$. If we order species according to the taxonomical tree (left panel), no particular structure is notable. If instead we reorder species using a dendrogram where pairwise similarities are measured by the posterior mean of $\Lambda \Lambda^\top$ (right panel), interesting patterns arise. For instance, there are blocks with positive pairwise dependence along the diagonal with mostly negative dependence off the blocks. This suggests the presence of groups of species that are not taxonomically closely related and have positive and negative interactions and/or relationships with latent environmental conditions. 

\begin{figure}
\centering
\begin{subfigure}{1\textwidth}
    \centering
    \includegraphics[width=0.6\linewidth]{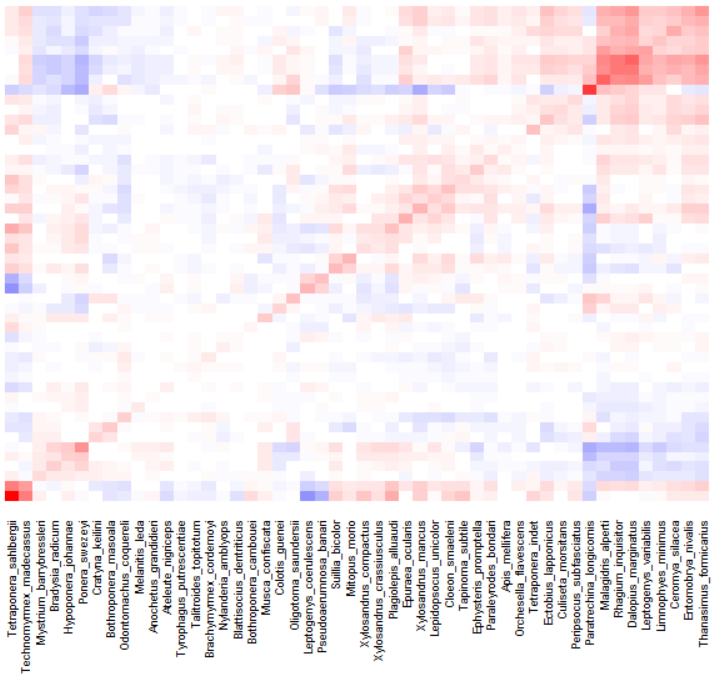}
\end{subfigure}
\vfill
\centering
\caption{Posterior mean of $\Lambda \Lambda^\top$ where entries for which the $95\%$ credible interval included $0$ were set to $0$ for 50 species ordered according to their pairwise latent covariance.}
\label{fig:insects_Lambda_50}
\end{figure}
Figure \ref{fig:insects_Lambda_50} zooms in on 50 species for which species' names are available. In the upper right corner we observe a group of species with positive pairwise dependence supporting again the existance of positive interactions of species which are not taxonomically close. For instance, this group includes various beetles belonging to the Coleoptera order (the Rhagium inquisitor, Dalopius marginatusn and Thanasimus formicarius) as well as species from different orders, such as Entomobrya nivalis, a species of slender springtails, and the Ceromya silacea, a species of fly.

\section{Discussion}\label{sec:conclusion}

We presented a method for fast estimation with accurate uncertainty quantification for multivariate logistic regression models with high-dimensional outcomes. There are several important directions for future research. An interesting avenue is to extend factor pre-estimation to any generalized linear latent variable model and to latent space models for network data \citep{hoff_02, durante_14}, while formalizing the theory on accuracy relative to the true posterior and frequentist coverage of the corresponding posterior approximations. In particular, it is worthwhile to develop a similar coverage correction strategy to bypass the need for expensive Gibbs sampling routines to quantify uncertainty, while allowing fast and accurate Bayesian inference in many important applied settings.

It is also interesting to increase flexibility by developing extensions to (1) non-linear latent factor models and (2) more complex and structured priors. 
For (1), an estimate of latent factors might be obtained via a suitable adaptation of a non-linear dimensionality reduction technique \citep{gpllvm, diffusion_map} and inference on the non-linear map could proceed adopting some non-parametric prior. Regarding (2), it is desirable to consider a hierarchical prior formulation, for instance incorporating phylogenetic information of species, shrinking regression coefficients and factor loadings of similar species towards a common estimate. This seems particularly important in our motivating ecological application, where we expect estimates for rare species to benefit from borrowing of information. 

Finally, applying \texttt{FLAIR} to other high-dimensional data sets measuring biodiversity
to assess the impact of climate and environmental disruption on species communities while uncovering interesting interactions between species is an important direction. Indeed, we expect \texttt{FLAIR} to transform practice in joint species distribution modeling of highly diverse groups, such as arthropods and fungi, since current methods fail to scale sufficiently to the sized datasets that are now being routinely collected. 

\section*{Acknowledgement}
This research was partially supported by the National Institutes of Health (grant ID R01ES035625), by the Office of Naval Research (Grant N00014-24-1-2626), by the European Research Council under the European Union’s Horizon 2020 research and innovation programme (grant agreement No 856506), and by Merck \& Co., Inc., through its support for the Merck Biostatistics and Research Decision Sciences (BARDS) Academic Collaboration. The authors thank Otso Ovaskainen and Shounak Chattopadhyay for helpful comments and the Lifeplan team for generously providing the motivating data.

\bibliographystyle{apalike}
\bibliography{references}

\newpage

\setcounter{page}{1}
\setcounter{section}{0}
\setcounter{table}{0}
\setcounter{figure}{0}
\setcounter{equation}{0}
\renewcommand{\theHsection}{SIsection.\arabic{section}}
\renewcommand{\theHtable}{SItable.\arabic{table}}
\renewcommand{\theHfigure}{SIfigure.\arabic{figure}}
\renewcommand{\thepage}{S\arabic{page}}  
\renewcommand{\thesection}{S\arabic{section}}   
\renewcommand{\theequation}{S\arabic{equation}}   
\renewcommand{\thetable}{S\arabic{table}}   
\renewcommand{\thefigure}{S\arabic{figure}}
\renewcommand{\thetheorem}{S\arabic{theorem}}
\renewcommand{\thelemma}{S\arabic{lemma}}
\renewcommand{\theproposition}{S\arabic{proposition}}
\renewcommand{\thecorollary}{S\arabic{corollary}}
\renewcommand{\theequation}{S\arabic{equation}}

\appendix
\section*{Supplementary Material of “Factor pre-training in Bayesian multivariate logistic models”}

\section{Proofs of the Main Results}
\begin{proof}[Proof of Theorem \ref{thm:recovery_Z}]
For a matrix $Z \in \mathbb R^{n \times p}$, where $Z = [z_{ij}]_{ij}$, we denote by $\mathcal L(Z) = \sum_{i=1}^n \sum_{j=1}^p y_{ij} \log\{h^{-1}(z_{ij})\} + (1-y_{ij})\log\{1-h^{-1}(z_{ij})\}$, for brevity.
    Then, for $(\hat M, \hat \Lambda, \hat B)$ the solution to \eqref{eq:cjmap}, consider the following decomposition \begin{equation*}
    \begin{aligned}
         \log p(\hat M, \hat \Lambda, \hat B\mid Y, X ) &- \log p\left ( M_0, \Lambda_0, B_0\mid Y, X \right) \\
         &= \mathcal L(\hat Z) - \mathcal L(Z_0) + \log p (\hat M) - \log p(M_0) + \log p (\hat \Lambda) - \log p(\Lambda_0) \\
         & \quad \quad + \log p(\hat B) - \log p(B_0)\\
         &= \mathcal L(\hat Z) - \mathcal L(Z_0) + \frac{1}{2}\big( \left|\left| M_0 \right|\right|_F^2-  \left|\left| \hat M\right|\right|_F^2\big) + \frac{1}{2}\big\{ tr(\Lambda_0^\top \Sigma_{\Lambda}^{-1} \Lambda_0)- tr(\hat \Lambda^\top \Sigma_{\Lambda}^{-1} \hat \Lambda)\big\} \\
         & \quad \quad +\frac{1}{2}\big\{tr(B_0^\top \Sigma_{B}^{-1} B_0)- tr(\hat B^\top \Sigma_{B}^{-1} \hat B)\big\}.
    \end{aligned}
\end{equation*}
Since, $ \log p(\hat M, \hat \Lambda, \hat B \mid Y, X) \geq \log p( M_0, \Lambda_0, B_0 \mid Y, X)$, we have
\begin{equation*}
\mathcal L(Z_0) - \mathcal L(\hat Z)   \leq \frac{1}{2} \big| || M_0 ||_F^2- || \hat M ||_F^2 +  tr(\hat \Lambda^\top \Sigma_{\Lambda}^{-1} \hat \Lambda) - tr( \Lambda_0^\top \Sigma_{\Lambda}^{-1}  \Lambda_0)  + tr(B_0^\top \Sigma_{B}^{-1} B_0)- tr(\hat B^\top \Sigma_{B}^{-1} \hat B)
 \big|.
\end{equation*}
First, we consider the event with high probability where the infinity norm of $M_0$ and $X$ can be suitably bounded. Define the events 
\begin{equation*}
\begin{aligned}
      A_1 &= \big\{||\eta_i||_{\infty} \leq 2 \log^{1/2}(kn), \ i=1, \dots, n\big\} = \big\{||M_0||_{\infty} \leq 2 \log^{1/2}(kn) \big\},\\
       A_2 &= \big\{||x_i||_{\infty} \leq c_x \log^{1/2}(qn), \ i=1, \dots, n\big\} = \big\{||X||_{\infty} \leq  c_x \log^{1/2}(qn) \big\}.
\end{aligned}  
\end{equation*}
By corollary \ref{corr:tail_max_gaus} and Assumption \ref{assumption:X}, we have $ \text{pr}(A_1) = \text{pr}\big\{||M_0||_{\infty} \leq 2 \log^{1/2}(kn) \big\} \geq 1- \frac{2}{n}$ and $\text{pr}(A_2) = 1 - C'/n$. 
Next, we restrict our analysis to the event $A_1 \cap A_2$. 
On the event $A_1 \cap A_2$, and under Assumptions \ref{assumption:lambda_sv}--\ref{assumption:X}, we have
\begin{equation*}
\left|\left|Z_0\right|\right|_{\infty} \leq \gamma_n, \quad \left|\left|\hat Z\right|\right|_{\infty} \leq \gamma_n, \quad \gamma_n = q c_B c_X \log^{1/2}(qn) +2 k c_\Lambda \log^{1/2}(kn)\lesssim \log^{1/2}\{(k \vee q) n\}
\end{equation*}
and
\begin{equation*}
    \begin{aligned}
        \big| || M_0 ||_F^2-  || \hat M||_F^2\big| &\leq  4 kn\log(kn),\\
           \big| tr(\Lambda_0^\top \Sigma_{\Lambda}^{-1} \Lambda_0)- tr(\hat \Lambda^\top \Sigma_{\Lambda}^{-1} \hat \Lambda)  \big| &\leq \frac{2}{\tau_{\Lambda, min}^2} c_\Lambda^2 kp,\\
            \big| tr(B_0^\top \Sigma_{B}^{-1} B_0)- tr(\hat B^\top \Sigma_{B}^{-1} \hat B)  \big| &\leq \frac{2}{\tau_{B, min}^2} c_B^2 qp,\\
    \end{aligned}
\end{equation*}
where $\tau_{\Lambda, min} = \min\left\{\tau_{\lambda_1}, \dots, \tau_{\lambda_p}\right\}$ and $\tau_{B, min} = \min\left\{\tau_{\beta_1}, \dots, \tau_{\beta_p}\right\}$.
Moreover, letting $b(z) = \log \left(1 + e^{z}\right)$, 
\begin{equation*}
    \begin{aligned}
        \mathcal L(\hat Z) - \mathcal L(Z_0) &= \sum_{i,j} \left[y_{ij}\left(\hat z_{ij} - z_{0ij}\right) - \left\{b\left(\hat z_{ij}\right) - b\left( z_{0ij}\right)\right\}\right] \\
        &= \sum_{i,j}  \left\{y_{ij} - b'\left(z_{0ij}\right)\right\}
        \left(\hat z_{ij} - z_{0ij}\right) - \left\{b\left(\hat z_{ij}\right) - b\left( z_{0ij}\right) -b'\left(z_{0ij}\right) 
        \left(\hat z_{ij} - z_{0ij}\right)  \right\} \\ 
        &=\sum_{i,j} \left\{y_{ij} - b'\left(z_{0ij}\right)\right\}\left(\hat z_{ij} - z_{0ij}\right) -  \frac{1}{2}b''\left(\delta_{ij}\right)\left(\hat z_{ij} - z_{0ij}\right)^2\\
       & \leq \sum_{i,j} \left\{y_{ij} - b'\left(z_{0ij}\right)\right\}\left(\hat z_{ij} - z_{0ij}\right) -\frac{1}{2} \inf_{|\delta| \leq \gamma_n }b''\left(\delta\right)\left(\hat z_{ij} - z_{0ij}\right)^2
    \end{aligned}
\end{equation*}
Hence,
\begin{equation}\label{eq:norm_sq_diff_Z}
    || \hat Z - Z_0||_F^2 \leq 2 \left\{\inf_{|\delta| \leq \gamma_n }b''\left(\delta\right)\right\}^{-1} \left[\sum_{i,j} \left\{y_{ij} -  b'\left(z_{0ij}\right) \right\}\left(\hat z_{ij} - z_{0ij}\right) - \mathcal L(\hat Z) - \mathcal L(Z_0)\right].
\end{equation}
Next, consider the following bound
\begin{equation*}
    \begin{aligned}
        \sum_{i,j} & \left\{y_{ij} - b'\left(z_{0ij}\right)\right\}\left(\hat z_{ij} - z_{0ij}\right) = tr\left\{\Psi^\top (\hat Z - Z_0)\right\} \leq  ||\Psi||_2  \left|\left| \hat Z - Z_0\right|\right|_F rank\left(\hat Z - Z_0\right),
    \end{aligned}
\end{equation*}
where $[\Psi]_{ij} = y_{ij} -  b'\left(z_{0ij}\right)$. 
Note that $b''(\delta) = \frac{e^{\delta}}{1 + e^{\delta}}  \frac{1}{1 + e^{\delta}}$, and $\inf_{|\delta| \leq \gamma_n }b''\left(\delta\right) \geq \frac{1}{2} \frac{1}{1+e^{\gamma_n}}$.
Hence, we get
\begin{equation}\label{eq:fr_ineq}
    \begin{aligned}
 \left|\left| \hat Z - Z_0\right|\right|_F^2 \leq 2 \left(1 + e^{\gamma_n}\right) \left\{\sqrt{2(k+q)}||\Psi||_2  \left|\left| \hat Z - Z_0\right|\right|_F +   2 kn\log(kn) + \frac{1}{\tau_{\Lambda, min}^2}c_\Lambda^2 kp +\frac{1}{\tau_{B, min}^2} c_B^2 qp\right\}.
    \end{aligned}
\end{equation}

Notice, we can rewrite \eqref{eq:fr_ineq} as a second-order inequality
\begin{equation}\label{eq:2_order_eq}
    x^2 - C_1 x - C_2 \leq 0,
\end{equation}
where $C_1 = 2 \left(1 + e^{\gamma_n}\right) \sqrt{2(k+q)}||\Psi||_2$, $C_2 = 2 \left(1 + e^{\gamma_n}\right) \left\{2 kn\log(kn) +  \frac{1}{\tau_{\Lambda, min}}c_\Lambda^2 kp +\frac{1}{\tau_{B, min}} c_B^2 qp\right\}$, and $x =  \left|\left| \hat Z - Z_0\right|\right|_F$. The positive root for \eqref{eq:2_order_eq} is given by
\begin{equation*}
    \begin{aligned}
        x &=\frac{1}{2} \left(C_1 + \sqrt{C_1^2 + 4C_2}\right ) \leq C_1 + \sqrt{C_2}\\
        & \leq  2 \left( 1 + e^{\gamma_n}\right)  \left\{\sqrt{2(k+q)}||\Psi||_2  + \sqrt{2 kn\log(kn)}   + \frac{1}{\tau_{\Lambda, min}^2}c_\Lambda \sqrt{kp} +\frac{1}{\tau_{B, min}^2} c_B \sqrt{qp}\right\}
    \end{aligned}
\end{equation*}
Finally, we bound $||\Psi||_2$. 
By Lemma \ref{lemma:psi}, with probability greater than $1 - 2/n$, we have
\begin{equation}\label{eq:s_1_Psi}
    \begin{aligned}
        s_1(\Psi) = ||\Psi||_2 \leq n^{1/2} + C_3 p^{1/2} + C_4\log^{1/2}(n) \lesssim n^{1/2} + p^{1/2}
    \end{aligned}
\end{equation}
where $C_3$ and $C_4$ are some constants not depending on $n$ and $p$.  We define $A_3$ as the event where \eqref{eq:s_1_Psi} holds.
Hence, under $A_1 \cap A_2 \cap A_3$, which has probability at least $1 - \text{pr}(A_1^c) - \text{pr}(A_2^c) - \text{pr}(A_3^c) = 1  - C/n$, where $C$ is an absolute constant,
\begin{equation*}
    \begin{aligned}
        \frac{1}{(np)^{1/2}} \left|\left| \hat Z - Z_0\right|\right|_F & \lesssim  e^{c_Z \log^{1/2}\{(k \vee q)n\}} \left\{\frac{1}{n^{1/2}} + \frac{\log^{1/2}(kn)}{p_n^{1/2}} \right\},
        \end{aligned}
\end{equation*} proving the result in \eqref{eq:parameters_cjmap}. 
Next, by Proposition \ref{prop:recovery_all}, \eqref{eq:recovery_Lambda_outer} and \eqref{eq:recovery_B} follow. \eqref{eq:posterior_mean_Lambda_outer_accuracy} follows from 
  \begin{equation*}
     \left|\left| \tilde\Sigma - \tilde\Lambda \tilde\Lambda^\top  \right| \right|_F = \rho^2 \left|\left|  D  \right| \right|_F \asymp p^{1/2}  e^{c_Z \log^{1/2}\{(k \vee q)n\}} \frac{1}{n},
  \end{equation*}
  combined with an application of the triangle inequality.
Finally, consider the following
    \begin{equation*}
         \left|\left|P_{\tilde U} - P_{U_0}\right|\right|_F^2 = 2k - 2 \textit{tr}\left(P_{\tilde U} P_{U_0}\right) = 2\left|\left| \sin{\left(\Theta_{\tilde U, U_0}\right)}\right|\right|_F^2,
    \end{equation*}
    where $\tilde U, U_0 \in \mathbb R^{n \times k}$ are the matrices of left singular vectors of $\tilde M \tilde \Lambda^\top$ and $M_0 \Lambda_0^\top$ respectively, $P_W$ denotes the orthogonal projection onto the column space of $W$, $P_W = W\left(W^\top W\right)^{-1}W^\top$, and $\sin{(\Theta_{\tilde U, U_0})}$ denote the sine of the angles between the subspaces spanned by $\tilde U$ and $U_0$.  
Moreover, by Theorem 20 in \citet{orourke18}, we have
\begin{equation*}
    \sin{(\Theta_{\tilde U, U_0})} \leq 2  \frac{  \left|\left|\tilde M \tilde \Lambda^\top - M_0 \Lambda_0^\top \right|\right|_2}{s_k\left( M_0 \Lambda_0^\top\right)}.
\end{equation*}
Note that $s_k\left( M_0 \Lambda_0^\top\right) \geq s_k\left( M_0 \right)s_k\left(\Lambda_0\right) \gtrsim (np)^{1/2}$ since $s_k\left( M_0 \right) \gtrsim C_5 n^{1/2}$ with probability $1 - o(1)$ by Lemma \ref{lemma:sv_gaussian}.
Hence, 
 \begin{equation*}
         \left|\left|P_{\tilde U} - P_{U_0}\right|\right|_F \lesssim \frac{\left|\left|\tilde M \tilde \Lambda^\top - M_0 \Lambda_0^\top \right|\right|_2}{(np)^{1/2}}
\end{equation*}
proving \eqref{eq:recovery_P_U}. By Davis-Kahan theorem \citep{davis_kahan} we have 
    \begin{equation*}
          \min_{R \in R^{k \times k}: R^\top R = I_k}||\tilde U - U_0 R ||_F = ||\tilde U - U_0 \hat R  ||_F \lesssim  ||P_{\tilde U} - P_{U_0}||_F 
    \end{equation*}
    where $\hat R$ achieves the minimum on the left hand side.
    Consider the singular value decomposition of $M_0$, $\bar U_0 D_0 V_0^\top$, where $\bar U_0 = U_0 \bar R^\top$, for some orthogonal matrix $\bar R \in \mathbb R^{k \times k}$. 
  Recalling that $\tilde M = \sqrt{n}\tilde U$ and letting $\tilde R = \mathbf{V}_0  \bar R \hat R$, we have \begin{equation*}
        \begin{aligned}
            ||\tilde M-  M_0 {\tilde R^\top}  || &= ||\sqrt{n} \tilde U - \bar U_0 D_0 V_0^\top {\tilde R} ||  = ||\sqrt{n} \tilde U -  U_0 \bar R^\top D_0 \bar R \hat R || \\ 
           & \leq  ||\sqrt{n} (\tilde U - U_0 {\hat R} ) || + ||\sqrt{n} U_0 {\hat R} -  U_0\bar R^\top D_0 \bar R {\hat R}  ||\\
            &\leq  ||\sqrt{n} (\tilde U - U_0 {\hat R} ) ||  + \max_{1 \leq l \leq k_0} |\sqrt{n} - d_{0l}|
        \end{aligned}
    \end{equation*}where $d_{0l}$ is the $l$-th largest singular value of $D_0$. Moreover, by corollary 5.35 of \citet{Vershynin_2012}, we have $|d_{0l} - \sqrt{n}| \lesssim \sqrt{k_0}$ with probability at least $1-o(1)$. The result follows from all of the above.
\end{proof}

\begin{remark}[Note on the bound of Theorem \ref{thm:recovery_Z}]
As discussed in the main paper, the bound in Theorem \ref{thm:recovery_Z} is less tight than similar bounds in related works. This is due to different assumptions on latent factors and covariates. Notice the presence of $\left\{\inf_{|\delta| \leq \gamma_n }b''\left(\delta\right)\right\}^{-1}$ in \eqref{eq:norm_sq_diff_Z} where $\gamma_n$ is an upper bound to the absolute value of the linear predictor. Similarly, the constant in (S.2) in the proof of Lemma 1 of \citet{lee24} contains the same factor, and Theorem 1 in \citet{davenport_1_bit} presents the same term for a similar bound. As $\gamma_n$ diverges, we have $\left\{\inf_{|\delta| \leq \gamma_n }b''\left(\delta\right)\right\}^{-1} \asymp e^{\gamma_n}$, while $\left\{\inf_{|\delta| \leq \gamma_n }b''\left(\delta\right)\right\}^{-1}$ is bounded for any finite value of $\gamma_n$.  Under the assumption $\eta_i \sim N_k (0, I_k)$, combined with our assumption on $X$ (Assumption \ref{assumption:X} of the main article), we can only bound the maximum of the absolute value of the linear predictor by some multiple of $\log^{1/2}\{(k \vee q) n\}$, which then determines the presence of the term $e^{c_z\log^{1/2}\{(k \vee q) n\}}$ in our results. \citet{lee24} assume covariates and latent factors fall in compact spaces that do not depend on $n$ and $p$ (Assumption 1 of \citet{lee24}). Consequently, they upper bound the absolute value of the linear predictor, and, in turn, upper and lower bound the second derivative of the log-partition function by a constant not dependent on $n$, which determines their sharper result. 
These assumptions in \citet{lee24} are restrictive in, for example, excluding cases in which covariates are independent samples from sub-Gaussian distributions. Moreover, even if the covariates are restricted to lie in some hypercube, assuming that latent factors are sampled as $\eta_i \sim N_k (0, I_k)$, as is common in the literature on random factor analysis \citep{West2003} and crucial for interpreting $\Lambda \Lambda^\top$ as latent covariance between outcomes, induces a factor of $e^{c_z\log^{1/2}(n)}$, since we can only bound the maximum of the absolute value of latent factors with some multiple of $\log^{1/2}(kn)$ with high probability (Lemma \ref{lemma:tail_max_subgassian}).
 Finally, we note that, for $n \lesssim p_n$, our bound would be asymptotically dominated by $n^{-1/2 + \epsilon}$, for any $\epsilon>0$, producing close to the parametric rate of convergence.   
 \end{remark}

\begin{proof}[Proof of Theorem \ref{thm:posterior_contraction}]
To prove posterior contraction, we show that the posterior distribution for $\Lambda \Lambda^\top$ and $B$ contract around $\tilde \Lambda \tilde \Lambda^\top$ and $\tilde B$ respectively.
Consider a sample for the posterior for $\Lambda$ and $B$. 
Due to the note on the posterior covariance (Section \ref{subsec:note_covariance}), we have
\begin{equation*}
\begin{aligned}
     ||\lambda_j - \tilde \lambda_j|| & \lesssim \frac{e^{c_Z /2\log^{1/2}\{(k \vee q)n\}}}{n^{1/2}}\log^{1/2}(p),\\
     ||\beta_j- \beta_j || & \lesssim \frac{e^{c_Z /2\log^{1/2}\{(k \vee q)n\}}}{n^{1/2}}\log^{1/2}(p)
\end{aligned}\quad j=1, \dots, p
\end{equation*}
with posterior probability at least $1 - 2/n$. 
Thus, with posterior probability at least $1 - 2/n$,
\begin{equation*}
\begin{aligned}
    \left|\left|\Lambda \Lambda^\top - \tilde \Lambda \tilde \Lambda^\top \right| \right|_F &\lesssim \frac{e^{c_Z /2\log^{1/2}\{(k \vee q)n\}}}{n^{1/2}}\log^{1/2}(p) p\\
\left|\left|B- \tilde B \right| \right|_F &\lesssim\frac{e^{c_Z /2\log^{1/2}\{(k \vee q)n\}}}{n^{1/2}}\log^{1/2}(p) (pq)^{1/2},
\end{aligned} \quad j=1, \dots, p,
\end{equation*}
where $\Lambda$ and $B$ are samples from $\tilde \Pi$.
An application of the triangle inequality combined with $\left|\left| \Lambda_0 \Lambda_0^\top\right| \right|_F \asymp p$ proves the result.
\end{proof}

\section{Auxiliary Results and Lemmas}
\subsection{Auxiliary Results} \label{subsec:additional_results}

\begin{proposition}
[Recovery of the factor analytic and linear predictor terms]
    \label{prop:recovery_M_Lambda}
    Define $\left(\tilde M, \tilde \Lambda, \tilde B\right)$ as the triplet obtained post-processing the joint maximum a posteriori estimate $\left(\hat M, \hat \Lambda, \hat B\right)$ solution to \eqref{eq:cjmap} via the procedure described in Section \ref{subsec:postprocessing}. Under the same assumption of Theorem \ref{thm:recovery_Z}, with probability at least $1-o(1)$, we have
    \begin{equation*}\label{eq:M_Lambda_cjmap}
    \begin{aligned}
    \frac{1}{(np_n)^{1/2}}\left|\left| \tilde M\tilde \Lambda^\top - M_0 \Lambda_0^\top \right|\right|_F &\leq \frac{\left|\left| \hat{Z} - Z_0 \right|\right|_F}{(np_n)^{1/2}} + C_1 \frac{kq}{n^{1/2}}\\
     \frac{1}{(np_n)^{1/2}}\left|\left| X \tilde B^\top - X B_0^\top\right|\right|_F &\leq \frac{\left|\left| \hat{Z} - Z_0 \right|\right|_F}{(np_n)^{1/2}} + C_2 \frac{kq}{n^{1/2}}\\
    \end{aligned}
    \end{equation*}
    where $C_1$ and $C_2$ are universal constants not depending on $n, p$.
\end{proposition}

\begin{proposition}
[Recovery of the regression coefficients matrix $B_0$]
    \label{prop:recovery_B}
    If Assumption \ref{assumption:X} holds, 
\begin{equation*}\label{eq:B_tilde_accuracy}
    \begin{aligned}
    \left|\left| \tilde B-  B_0 \right|\right|_F &\lesssim \frac{1}{n^{1/2}}
     \left|\left| X \tilde B^\top - X B_0^\top\right|\right|_F .
    \end{aligned}
    \end{equation*}
\end{proposition}

\begin{proposition}\label{prop:accuracy_Lambda_outer}
If $\tilde M$ is such that $\tilde M^\top \tilde M = n I_k$,   Assumptions \ref{assumption:lambda_sv}, \ref{assumption:B}, \ref{assumption:X} hold, and
 \begin{equation*} 
    \frac{1}{(np_n)^{1/2}}\left|\left|\tilde M \tilde \Lambda^\top - M_0 \Lambda_0 \right|\right|  \lesssim \delta_n
\end{equation*}
with $\delta_n \to 0$ and $\delta_n \gtrsim n^{-1/2}$. Then, with probability at least $1- o(1)$, we have
\begin{equation*}
    \begin{aligned}
            \frac{ \left|\left|  \tilde \Lambda \tilde \Lambda^\top -   \Lambda_0 \Lambda_0^\top \right|\right|_F}{\left|\left|  \Lambda_0 \Lambda_0^\top \right|\right|_F} &\lesssim \delta_n.
    \end{aligned}
\end{equation*}
\end{proposition}

\begin{proposition}\label{prop:recovery_all}
    If $\tilde M$ is such that $\tilde M^\top \tilde M = n I_k$, Assumptions \ref{assumption:lambda_sv}, \ref{assumption:B}, \ref{assumption:X} hold, \begin{equation*} 
    \frac{1}{(np_n)^{1/2}}\left|\left|\tilde M \tilde \Lambda^\top - M_0 \Lambda_0 \right|\right|  \lesssim \delta_n
\end{equation*}
with $\delta_n \to 0$ and $\delta_n \gtrsim n^{-1/2}$. Then,  with probability at least $1- o(1)$, 
\begin{equation*}
    \begin{aligned}
          \frac{\left|\left|  \tilde \Lambda \tilde \Lambda^\top -   \Lambda_0 \Lambda_0^\top \right|\right|_F }{\left|\left|  \Lambda_0 \Lambda_0^\top \right|\right|_F }&\lesssim \delta_n, \quad  
           \frac{1}{(pq)^{1/2}}\left|\left| \tilde B-  B_0 \right|\right|_F &\lesssim \delta_n.
    \end{aligned}
\end{equation*}
\end{proposition}

\subsection{Proofs of Auxiliary Results}

\begin{proof}[Proof of Proposition \ref{prop:recovery_M_Lambda}]
    First, notice $\tilde M \tilde \Lambda^\top = (I-P_X) \hat Z$ and $M_0 \Lambda_0^\top = (I-P_X) Z_0 + P_XM_0\Lambda_0^\top$, where $P_X = X\left(X^\top X\right)^{-1}X$. Similarly, $X \tilde B^\top = P_X \hat Z$ and $X B_0^\top = P_X Z_0 - P_X M_0 \Lambda_0^\top$. Thus, we have
\begin{equation*}
    \begin{aligned}
       \left|\left|\tilde M \tilde \Lambda^\top  -M_0\Lambda_0^\top\right|\right|_F & =   \left|\left| (I-P_X) \hat Z  - (I-P_X) Z_0 + P_X M_0 \Lambda_0^\top \right|\right|_F \\
       &\leq \left|\left| (I-P_X) \left(\hat Z  -Z_0 \right)\right|\right|_F  +\left|\left| P_X M_0 \Lambda_0^\top \right|\right|_F \\
        & \leq \left|\left| \hat Z  -Z_0 \right|\right|_F  +\left|\left| P_X M_0 \Lambda_0^\top \right|\right|_F 
    \end{aligned}
\end{equation*}
Moreover, consider $P_X = U_X U_X^\top$, where $U_X \in \mathbb R^{n \times q}$ and $U_X^\top U_X = I_q$. Then, $\left|\left| P_X M_0 \Lambda_0^\top \right|\right|_F  = \left|\left| U_X^\top M_0 \Lambda_0^\top \right|\right|_F$ and elements of 
$U_X^\top M_0$ are independent standard normal random variables. Hence, we have
$\left|\left| U_X^\top M_0 \Lambda_0^\top \right|\right|_F \leq \left|\left| U_X^\top M_0  \right|\right|_F \left|\left|\Lambda_0 \right|\right|_F \lesssim kqp^{1/2}$ with probability $1- o(1)$, since $\left|\left| U_X^\top M_0  \right|\right|_F \lesssim \sqrt{kq}$ with probability $1- o(1)$ by Lemma \ref{lemma:tail_norm_gaussian}.
With similar steps, we can obtain
\begin{equation*}
    \begin{aligned}
        \frac{1}{(np)^{1/2}}\left|\left|X \tilde B^\top - X B_0^\top \right|\right|_F & \lesssim \frac{1}{(np)^{1/2}}\left|\left| \hat Z  -Z_0 \right|\right|_F  + \frac{k}{n^{1/2}}.
    \end{aligned}
\end{equation*}
\end{proof}

\begin{proof}[Proof of Proposition \ref{prop:recovery_B}]The result follows from
    \begin{equation*}\label{eq:B_tilde_accuracy}
    \begin{aligned}
    \left|\left| \tilde B-  B_0 \right|\right|_F &=
     \left|\left|(X^\top X)^{-1}X^\top \left( X \tilde B^\top - X B_0^\top\right)\right|\right|_F \leq  \left|\left|(X^\top X)^{-1}X^\top \right|\right|_F \left|\left| X \tilde B^\top - X B_0^\top\right|\right|_F \\
     &\leq \left|\left|(X^\top X)^{-1}\right|\right|_F \left|\left|X  \right|\right|_F  \left|\left| X \tilde B^\top - X B_0^\top\right|\right|_F\\
     & \lesssim  \frac{k}{n^{1/2}} \left|\left| X \tilde B^\top - X B_0^\top\right|\right|_F,
    \end{aligned}
    \end{equation*}
    where the last inequality follows from Assumption \ref{assumption:X}.
\end{proof}

\begin{proof}[Proof of Proposition \ref{prop:accuracy_Lambda_outer}]

Define $E = \tilde M \tilde \Lambda^\top - M_0 \Lambda_0^\top$. Then,
\begin{equation*}
    \begin{aligned}
        \tilde \Lambda \tilde \Lambda^\top &= \frac{1}{n}  \tilde \Lambda  \tilde M^\top  \tilde M  \tilde \Lambda^\top = \frac{1}{n} \left(\Lambda_0 M_0^\top M_0 \Lambda_0^\top + E^\top E + E^\top M_0 \Lambda_0^\top + \Lambda_0 M_0^\top E \right)
    \end{aligned}
\end{equation*}
and
\begin{equation*}
    \Lambda_0 \Lambda_0^\top  =\frac{1}{n} \Lambda_0 M_0^\top M_0 \Lambda_0^\top  + \Lambda_0 \left(I - \frac{1}{n} M_0^\top M_0\right) \Lambda_0^\top 
\end{equation*}
Moreover, 
\begin{equation*}
    \begin{aligned}
        \left|\left| E^\top E \right|\right|_F &\leq \left|\left| E \right|\right|_F^2 \lesssim \delta_n^2 np, \\%
         \left|\left| E^\top M_0 \Lambda_0^\top \right|\right|_F &\leq   \left|\left| E\right|\right|_F  \left|\left| M_0 \right|\right|_2  \left|\left|\Lambda_0^\top \right|\right|_F \lesssim \delta_n np,\\
         \left|\left| \Lambda_0 \left(I - \frac{1}{n} M_0^\top M_0\right) \Lambda_0^\top \right|\right|_F &\leq  \left|\left| \Lambda_0  \right|\right|_F^2  \left|\left|I - \frac{1}{n} M_0^\top M_0\right|\right|_2 \lesssim \left|\left| \Lambda_0  \right|\right|_F^2 \frac{\sqrt{k}}{n^{1/2}}
    \end{aligned}
\end{equation*}
since, by Lemma \ref{lemma:sv_gaussian}, with probability $1  - o(1)$, we have $||M_0||_2 \lesssim n^{1/2}$ and $\left|\left|I - \frac{1}{n} M_0^\top M_0\right|\right|_2  \lesssim \frac{k}{n^{1/2}}$.
Thus,
\begin{equation*}
    \begin{aligned}
          \left|\left|  \tilde \Lambda \tilde \Lambda^\top -   \Lambda_0 \Lambda_0^\top \right|\right|_F \lesssim
          \delta_n^2 p +  \delta_n p + \left|\left| \Lambda_0  \right|\right|_F^2 \frac{\sqrt{k}}{n^{1/2}}.
    \end{aligned}
\end{equation*}
The result follows from $\left|\left|\Lambda_0 \Lambda_0^\top\right|\right|_2 \asymp p$.

\end{proof}

\begin{proof}[Proof of Proposition \ref{prop:recovery_all}]
    Follows from Proposition \ref{prop:recovery_M_Lambda}, \ref{prop:recovery_B} and \ref{prop:accuracy_Lambda_outer}.
\end{proof}

\subsection{Auxiliary Lemmas}

\begin{lemma}[Tail Probability of the Maximum of Sub-Gaussian Random Variables]\label{lemma:tail_max_subgassian}
Let $X_i$ be independent and identically distributed for $i=1, \dots, n$ $\sigma^2$-sub-Gaussian random variables.
Then,
\begin{equation*}
    \text{pr} \Big[ \max_{i=1,\dots, n} X_i > \left[2\sigma^2\left\{\log (n) +t \right\}\right]^{1/2}\Big]\leq e^{-t}
\end{equation*}
\end{lemma}
\begin{proof}[Proof Lemma \ref{lemma:tail_max_subgassian}]
The result follows from
    \begin{equation*}
          \text{pr} \Big(\max_{i=1,\dots, n} X_i > u\Big) \leq \sum_{i=1}^n \text{pr} (X_i > u) \leq n e^{-\frac{u^2}{2\sigma^2}}.
    \end{equation*}
\end{proof}
\begin{corollary}[Corollary of Lemma \ref{lemma:tail_max_subgassian}]\label{corr:tail_max_gaus}
    Letting $M_0 \in \mathbb R^{n \times k}$, with $[M_0]_{ij} \sim N(0,1)$ independent,
    \begin{equation*}
         \text{pr} \left[||M_0||_{\infty} > \{2\log(kn) + 2\log(n)\}^{1/2} \right]\leq 2e^{-\log(n)} = \frac{2}{n}.
    \end{equation*}
    \end{corollary}

\begin{lemma}[Lemma 1 of \citet{laurent_massart} (Tail Probability of the Norm of a Gaussian Vector)]\label{lemma:tail_norm_gaussian}
Consider $X \sim N_p(0,  \Sigma)$, then
    \begin{equation*}
        \text{pr} \left(||X||_2^2 > tr(\Sigma) +2\sqrt{t}||\Sigma||_F +2t ||\Sigma||_2\right) \leq e^{-t}.
    \end{equation*}
Hence, for $X \sim N_p(0, \sigma^2 I_p)$, then
    \begin{equation*}
        \text{pr} \left(||X||_2 > t\right) \leq 2 \exp\left(- \frac{t^2}{2p \sigma^2}\right).
    \end{equation*}
\end{lemma}
\begin{corollary} [Corollary of Lemma \ref{lemma:tail_norm_gaussian}]
If $\eta_i \sim N_k(0, I_k)$ independently, then
\begin{equation*}
     \text{pr} \bigg\{\bigg|\bigg|\sum_{i=1}^n \eta_i \bigg|\bigg|_2 > (2nk)^{1/2}\log^{1/2}(n) \bigg\} \leq \frac{2}{n}.
\end{equation*}

\end{corollary}

\begin{lemma}[Singular Values of Matrix with independent Gaussian Entries]\label{lemma:sv_gaussian}
Consider a matrix $X \in \mathbb R^{n \times k}$, such that $[X]_{ij} \sim N(0,1)$ independently, then
    \begin{equation*}
        \text{pr} \left\{n^{1/2} -k^{1/2} -t \leq s_d(X) \leq s_1(X) \leq n^{1/2} +k^{1/2} + t\right\} \geq 1 - 2e^{-t^2/2}
    \end{equation*}
\end{lemma}
\begin{proof}[Proof of Lemma \ref{lemma:sv_gaussian}]
    See chapter 1 of \citet{vershynin_08}.
\end{proof}

\begin{lemma}[Singular values of $\Psi$]\label{lemma:psi}
Define the matrix $\Psi \in \mathbb R^{n \times p}$, where $[\Psi]_{ij} = y_{ij} - h(x_i^\top \beta_j + \eta_i^\top \lambda_j)$, where $\eta_i \sim N_k(0, I_k)$ independently. Then, with probability at least $1-2e^{-ct^2}$,
\begin{equation*}
    s_1(\Psi) \leq n^{1/2} + C p^{1/2} +t,
\end{equation*}
where $c$ and $C$ are absolute constants.
\end{lemma}

\begin{proof}[Proof of Lemma \ref{lemma:psi}]
We modify the proof of Theorem 5.39 in \citet{Vershynin_2012} to matrices with independent rows and  non-common diagonal second moment. 
    We first condition on the realization of $M_0$ and consider it fixed. Next, we derive the conclusion since the desired result holds for every $M_0$.
     For $i=1, \dots, n$ and $j=1, \dots, p$, define $p_{ij} = h(x_i^\top \beta_j + \eta_i^\top \lambda_j)$ and $v_{ij} = p_{ij}(1- p_{ij})$, where, for simplicity, we dropped the dependence on the $\eta_i$'s. Recall that conditionally on $M_0$, the elements of $\Psi$ are independent, and
    \begin{equation*}
        E(\Psi_{ij} \mid M_0) = 0, \quad  var(\Psi_{ij} \mid M_0) = v_{ij}.
    \end{equation*}
    For $i=1, \dots, n$, define $V_i = diag\left(v_{i1}, \dots, v_{ip}\right)$ and $\bar V = \frac{1}{n} \sum_{i=1}^n V_i$.
    Our conclusion is equivalent to showing 
\begin{equation}\label{eq:norm_equivalent}
\left|\left| \frac{1}{n} \Psi^\top \Psi - \bar V \right| \right| \leq \max(\delta, \delta^2) = \epsilon, \quad \delta= C \frac{p^{1/2}}{n^{1/2}} + \frac{t}{n^{1/2}},   
    \end{equation} with high probability.
Indeed, if \eqref{eq:norm_equivalent} holds, then for any $x \in S^{p-1}$, where $S^{p-1}$ denotes the unit sphere in $R^p$, $\left|\left| \frac{1}{n^{1/2}}\Psi x\right| \right|^2\leq \epsilon + |x^\top \bar V x| \leq \epsilon + \frac{1}{4}$, which implies $\left|\left| \frac{1}{n^{1/2}}\Psi x\right| \right|\leq \delta + \frac{1}{4}$, and, consequently, $s_1(\Psi) \leq \frac{n}{4} +  C p^{1/2} +t$.
Denote by $N$ a $1/4$-net of $S^{p-1}$, then, by Lemma 5.4 in \citet{Vershynin_2012}, we have
\begin{equation*}
    \left|\left| \frac{1}{n} \Psi^\top \Psi - \bar V \right| \right| \leq 2 \max_{x \in N} \left|\left\langle \left( \frac{1}{n} \Psi^\top \Psi - \bar V\right)x ,x\right\rangle \right| = 2 \max_{x \in N} \left| \frac{1}{n}\left|\left| \Psi x\right| \right|^2 - x^\top\bar Vx  \right|
\end{equation*}
It remains to show $\max_{x \in N} \left| \frac{1}{n}\left|\left| \Psi x\right| \right|^2 - x^\top\bar Vx  \right| \leq \frac{\epsilon}{2}$.
By Lemma 5.2 in \citet{Vershynin_2012}, we can choose $N$ to be of cardinality at most $9^p$.
Fix a vector $x \in S^{p-1}$, and define $Z_i = \Psi_i x$, where $\Psi_i^\top$ is the $i$-th row of $\Psi$, then $||\Psi x||^2 = \sum_{i}^{n} Z_i^2$.  The $Z_i$'s are independent, sub-Gaussian random variables with $E(Z_i^2) = x^\top V_i x \leq 1/4$ and $||Z_i||_{\psi} \leq \max_{j}||\Psi_{ij}|| \leq \frac{1}{\log(2)}$, where $||X||_{\psi}$ denotes the sub-Gaussian norm of $X$.
Hence, we have
\begin{equation*}
    \text{pr}\left( \left|\frac{1}{n}\sum_{i=1}^n Z_i^2 - x^\top \bar{V} x\right| >\frac{\epsilon}{2}\right) \leq 2 \exp\left\{-c_1 \min(\epsilon, \epsilon^2)N\right\} \leq 2 \exp\left(-c_1 \delta^2 N\right) \leq  2 \exp\left\{-c_1 \left(C^2p + t^2\right)\right\}
\end{equation*}
where the first inequality follows from Corollary 5.17 in \citet{Vershynin_2012}, and $c_1 = \frac{1}{32e^2 \log(2)}$. Thus,
\begin{equation*}
    \text{pr}\left( \max_{x \in N} \left| \frac{1}{n}\left|\left| \Psi x\right| \right|^2 - x^\top\bar Vx  \right| >\frac{\epsilon}{2}\right) \leq 9^p 2\exp\left\{-c_1 \left(C^2p + t^2\right)\right\} \leq 2 \exp\left(-c_1 t^2\right)
\end{equation*}
where the last inequality follows from choosing $C = \sqrt{\frac{\log(9)}{c_1}}$.
\end{proof}

\subsection{Note on the Posterior Variance}\label{subsec:note_covariance}

Recall the posterior variance for $\theta_j$ is given by $\rho^2 \tilde V_j$, where
\begin{equation*}
    \tilde V_j = \left\{-\frac{\partial^2}{\partial \theta_j \partial \theta_j^\top}p( Y^{(j)} \mid X, \tilde M, \lambda_j, \beta_j) + \log p(\lambda_j \mid \hat \tau_{\lambda_j}) + \log p(\beta_j \mid \hat \tau_{\beta_j}) \mid_{\theta_j = \tilde \theta_j} \right\}^{-1}
\end{equation*}
In the following, we assume $\tau_{\beta_j} = \mathcal O(1)$,  $\tau_{\lambda_j} = \mathcal O(1)$,  and $\rho = \mathcal O(1)$. 
In particular, we have
\begin{equation*}
    \tilde V_j^{-1} = \begin{bmatrix}
        A + \tau_{\beta_j}^{-2}I_q & B\\
        B^\top & C+ \tau_{\lambda_j}^{-2}I_k
    \end{bmatrix},
\end{equation*}
where
\begin{equation*}
    \begin{aligned}
        A & =  \sum_{i=1}^n p_{ij} (1-p_{ij})  x_i x_i^\top = X^\top W_j X,\\
        B &= \sum_{i=1}^n p_{ij} (1-p_{ij})  x_i \tilde \eta_i^\top = X^\top W_j \tilde M,\\
        C & = \sum_{i=1}^n p_{ij} (1-p_{ij}) \tilde \eta_i \tilde \eta_i^\top= \tilde M^\top W_j  \tilde M,
    \end{aligned}
\end{equation*}
and $p_{ij} = h(\tilde z_{ij}) = \frac{1}{1+e^{-z_{ij}}}$, $W_j=\text{diag}(w_{j1}, \dots, w_{jn} )$, and $w_{ji} = p_{ij} (1-p_{ij})$.
Hence, \begin{equation*}
    V_j^{-1} = [X ~ \tilde{M}]^\top W_j [X ~ \tilde{M}] + \begin{bmatrix}
        \tau_{\beta_j}^{-2}I_q & 0\\
        0 &\tau_{\lambda_j}^{-2}I_k
    \end{bmatrix},
\end{equation*}
Moreover,
\begin{equation*}
   [X ~ \tilde{M}]^\top W_j [X ~ \tilde{M}]\succeq  w_{j,min} [X ~ \tilde{M}]^\top [X ~ \tilde{M}] = \begin{bmatrix}
       X^\top X & 0\\
       0 & n I_k
   \end{bmatrix},
\end{equation*}
where $w_{j,min} = \min\{w_{j1}, \dots, w_{jn}\}$, $w_{j,min} \geq \frac{1}{2} h(-\gamma_n) = \frac{1}{2} \frac{1}{1 + e^{\gamma_n}}$, and $\gamma_n \lesssim  \log^{1/2}\{(k \vee q)n\}$ is an upper bound to $\tilde z_{ij}$, and the equality follows from $\tilde M = n^{1/2} \tilde U$ with $\tilde U^\top \tilde U = I_k$ and $\tilde M^\top X = 0$.
Thus,
\begin{equation*}
    \begin{aligned}
        V_j  \preceq 2 (1 + e^{\gamma_n}) \begin{bmatrix}
            (X^\top X)^{-1} & 0\\
            0 &  \frac{1}{n} I_k
        \end{bmatrix} 
    \end{aligned}.
\end{equation*}
Recall that under Assumption \ref{assumption:X}, we have $(X^\top X)^{-1} \preceq \frac{C}{n}I_q$, where $C$ is a universal constant. 
This implies that for a sample $ \theta^{(s)} = \left(\theta_1^{(s)}, \dots,\theta_p^{(s)} \right)$ from $\tilde \Pi$ we have $\theta_j^{(s)} \overset{d}{=} \tilde \theta_j + \rho \tilde V_j^{1/2} \nu_j$, with $\nu_j \sim N_{k+q}(0, I_{k+q})$, where $\overset{d}{=}$ implies equality in distribution.
Hence, we have
\begin{equation}
    \left| \left| \theta_j^{s}  - \tilde\theta_j \right| \right|_F \lesssim e^{\gamma_n/2} \frac{\log(p)}{n^{1/2}}, 
\end{equation}
for all $j=1, \dots, p$, with posterior probability at least $1 - o(1)$.

\section{Extension to the Probit Model}\label{sec:extension_to_probit}
\subsection{Main Result}
It is interesting to extend the results above to other models for binary data, for instance using the probit link $\Phi^{-1}(\cdot)$, where $\Phi(\cdot)$ denotes the cumulative distribution function of a standard normal random variable.
The following Theorem shows that the joint maximum a posteriori estimates obtained under a probit link have asymptotic accuracy guarantees in approximating the true sample and outcome-specific probabilities in the large $p$ and $n$ regime.   
\begin{theorem}[Recovery of the distribution under the probit model]\label{thm:recovery_hellinger}
    Suppose Assumptions \ref{assumption:dgm} -- \ref{assumption:X} hold with $h(\cdot)$ replaced by $\Phi(\cdot)$ in equation \eqref{eq:gllvm}. Define $\hat Z = X \hat B^\top + \hat M \hat \Lambda^\top$, where $\left(\hat M, \hat \Lambda, \hat B\right)$ is the solution to \eqref{eq:cjmap}, with the probit likelihood replacing the logistic one, then, with probability at least $1- C/n $
    \begin{equation*}\label{eq:hellinger_cjmle}
       d_H^2\{\Phi \left(\hat Z\right), \Phi \left(Z_0\right)\} \lesssim  \log \{(k \vee q) n\}\left(\frac{1}{n^{1/2}} + \frac{1}{p_n^{1/2}}\right),
    \end{equation*}
    where $d_H^2\left\{\Phi \left(\hat Z\right), \Phi \left(Z_0\right)\right\} = \frac{1}{np}\sum_{i=1}^n \sum_{j=1}^{p_n} d_H^2\left\{\Phi \left(\hat z_{ij}\right), \Phi \left(z_{0ij}\right)\right\}$,  $d_H^2(f, g) = \left(\sqrt{f}-\sqrt{g}\right)^2 + \left(\sqrt{1-f}-\sqrt{1-g}\right)^2 $, and $C$ is a universal constant not depending on $n$ and $p_n$. 
\end{theorem}

\begin{remark}
    The same bound can be derived for the Kullback–Leibler  divergence $\mathcal D\left\{\Phi \left(Z\right) \mid \mid \Phi \left(\hat Z\right)\right\} = \frac{1}{np}\sum_{i=1}^n \sum_{j=1}^p \left[\Phi(z_{ij}) \log\left\{\frac{\Phi(\hat z_{ij})}{\Phi(z_{ij})}\right\} + \left\{1-\Phi(z_{ij})\right\}\log\left\{\frac{1-\Phi(\hat z_{ij})}{1-\Phi(z_{ij})}\right\}\right]$ and the squared total variation distance    \\ $d_{TV}^2\left\{\Phi \left(\hat Z\right), \Phi \left(Z\right)\right\} = \frac{1}{np}\sum_{i=1}^n \sum_{j=1}^p\left|\Phi(\hat z_{ij}) - \Phi(z_{ij}) \right|^2$.
\end{remark}

It would be appealing to modify 
the result in Theorem \ref{thm:recovery_hellinger}
to bound the norm of $\hat Z - Z$, as in \eqref{eq:parameters_cjmap}.  However, this is not trivial due to the flatness of the probit likelihood in the tails.
 However, our preliminary numerical results show that our method performs extremely well in the probit case; formally justifying this performance theoretically including for broader classes of link functions is an interesting area for future research.

\subsection{Proof  of Theorem \ref{thm:recovery_hellinger}}
\begin{proof}[Proof of Theorem \ref{thm:recovery_hellinger}]
We follow the proof of Theorem 2 in \citet{davenport_1_bit} with three modifications: firstly, $M_0$, and, hence, $Z_0$ are not fixed but random, secondly, we consider a general design matrix $X$, and, thirdly, we consider the joint maximum a posteriori estimate under  truncated Gaussian priors instead of the joint maximum likelihood estimate. 
Consider the difference between the log-posterior computed at the joint maximum a posteriori estimate and true parameter respectively.
For a matrix $Z \in \mathbb R^{n \times p}$, where $[Z]_{ij} = z_{ij}$, with a slight abuse of notation, we redefine  $\mathcal L(Z) = \sum_{i=1}^n \sum_{j=1}^p y_{ij} \log\left\{\Phi(z_{ij})\right\} + (1-y_{ij})\log\left\{1-\Phi(z_{ij})\right\}$, which is the log-likelihood under the probit link function. Recall the decomposition \begin{equation*}
    \begin{aligned}
         \log p(\hat M, \hat \Lambda, \hat B \mid Y, X) &- \log p\left ( M_0, \Lambda_0, B_0 \mid Y, X \right) \\
         &= \mathcal L(\hat Z) - \mathcal L(Z_0) + \log p (\hat M) - \log p(M_0) + \log p(\hat \Lambda) - \log p(\Lambda_0)\\
         & \quad  \quad + \log p(\hat B) - \log p(B_0)\\
         &= \mathcal L(\hat Z) - \mathcal L(Z_0) + \frac{1}{2}\left( \left|\left| M_0 \right|\right|_F^2-  \left|\left| \hat M\right|\right|_F^2\right) + \frac{1}{2}\left\{ tr\left(\Lambda_0^\top \Sigma_{\Lambda}^{-1} \Lambda_0\right)- tr\left(\hat \Lambda^\top \Sigma_{\Lambda}^{-1} \hat \Lambda\right)\right\} \\
         & \quad \quad +\frac{1}{2}\left\{tr\left(B_0^\top \Sigma_{B}^{-1} B_0\right)- tr\left(\hat B^\top \Sigma_{B}^{-1} \hat B\right)\right\}.
    \end{aligned}
\end{equation*}
With the same steps of the Proof for Theorem \ref{thm:recovery_Z}, we obtain 
\begin{equation*}
    \begin{aligned}
      \mathcal L( Z_0) -   \mathcal L(\tilde Z) \lesssim n \log(kn) + (k+q)p
    \end{aligned}
\end{equation*}
with probability at least $1 - C/n$ for some absolute constant $C$. 
Define  $\mathcal {\bar L}( Z) = \mathcal {L}( Z) - \mathcal { L}( 0)$, and consider the following expectation
\begin{equation*}
    \begin{aligned}
       E\left(\mathcal L(\hat Z) - \mathcal L(Z_0) \mid M_0  \right) &=  E\left(\bar{\mathcal L}(\hat Z) - \bar{\mathcal L}(Z_0) \mid M_0  \right)   \\ &=
        \sum_{i=1}^n \sum_{j=1}^p E\left(y_{ij} \log\left\{\frac{\Phi(\hat z_{ij})}{\Phi(\hat z_{0ij})}\right\} - (1 - y_{ij}) \log\left\{\frac{1-\Phi(\hat z_{ij})}{1-\Phi(\hat z_{0ij})}\right\}\mid M_0  \right) \\
        &=\sum_{i=1}^n \sum_{j=1}^p  \left[ \Phi( z_{0ij}) \log\left\{\frac{\Phi(\hat z_{ij})}{\Phi(\hat z_{0ij})}\right\}  - \{1 - \Phi(\hat z_{0ij})\} \log\left\{\frac{1-\Phi( z_{ij})}{1-\Phi(\hat z_{0ij})}\right\}\right]\\
        &= - np \mathcal{D}\left\{\Phi\left(Z_0\right) \mid \mid \Phi\left(\hat Z\right)\right\},
    \end{aligned}
\end{equation*}
where $\mathcal{D} \left\{\Phi\left(Z_0\right) \mid \mid \Phi\left(\hat Z\right)\right\} = \frac{1}{np} \sum_{i,j}\Phi( z_{0ij}) \log\left\{\frac{\Phi(\hat z_{ij})}{\Phi(\hat z_{0ij})}\right\}  - \{1 - \Phi(\hat z_{0ij})\} \log\left\{\frac{1-\Phi( z_{ij})}{1-\Phi(\hat z_{0ij})}\right\}$ denotes the average KL divergence across rows and columns.
Next, consider the following decomposition
\begin{equation*}
    \begin{aligned}
        \mathcal L(\hat Z) - \mathcal L(Z_0) &=  E\left(\mathcal L(\hat Z) - \mathcal L(Z_0) \mid M_0  \right) +  \mathcal L(\hat Z) - E\left\{\mathcal L(\hat Z) \mid M_0  \right\} - \left[\mathcal L( Z_0) - E\left(\mathcal L( Z_0) \mid M_0  \right) \right]\\
        &\leq  - np \mathcal{D}\left\{\Phi\left(Z_0\right) \mid \mid \Phi\left(\hat Z\right)\right\} + 2 \sup_{Z \in \mathcal G}\left| \mathcal L( Z)- E\left(\mathcal L( Z) \mid M_0  \right) \right|.
    \end{aligned}
\end{equation*}
where 
    $\mathcal{G} =  \left\{Z :~ \text{rank}(Z) = k+q \quad ||Z||_{\infty} \leq \gamma_n \right \}$,
since $Z_0, \hat Z \in \mathcal G$.
Combining all the above, we have
\begin{equation*}
    \begin{aligned}
         \log p(\hat \mu, \hat M, \hat \Lambda\mid Y) &- \log p\left ( \mu_0, M_0, \Lambda_0\mid Y\right ) \\
         &\leq - np \mathcal{D}\left\{\Phi\left(Z_0\right) \mid \mid \Phi\left(\hat Z\right)\right\} + 2 \sup_{Z \in \mathcal G}\left| \mathcal L( Z)- E\left(\mathcal L( Z) \mid M_0  \right) \right| +C_2\left\{kn\log(kn)+ (k+q)p\right\} 
    \end{aligned}
\end{equation*}
for some absolute constant $C_2$, and, since $\log p(\hat \mu, \hat M, \hat \Lambda\mid Y) - \log p\left ( \mu_0, M_0, \Lambda_0\mid Y\right ) \geq 0$, we obtain
\begin{equation}\label{eq:D_cjmap_1}
    \begin{aligned}
        \mathcal{D}\left\{\Phi\left(Z_0\right) \mid \mid \Phi\left(\hat Z\right)\right\} \leq \frac{1}{np} 2 \sup_{Z \in \mathcal G}\left| \mathcal L( Z)- E\left(\mathcal L( Z) \mid M_0  \right) \right| +C_2\left\{kn\log(kn)+ (k+q)p\right\}.  
    \end{aligned}
\end{equation}
To bound the first term on the right hand side of \eqref{eq:D_cjmap_1}, we rely on Lemma \ref{lemma:sup_cjmap}. In particular, we first define $A_{4} = \left\{\sup_{Z \in \mathcal G}\left| \mathcal L( Z)- E\left\{\mathcal L( Z)  \right\} \right| \geq C_0   \log\{(k \vee q)n\}\{np(k+q)\}^{1/2}\left( n^{1/2} + p^{1/2}\right)\right\}$. By Lemma \ref{lemma:sup_cjmap}, conditionally on the realization of $M_0$, we have $$\text{pr} \left[\sup_{Z \in \mathcal G}\left| \mathcal L( Z)- E\left(\mathcal L( Z) \mid M_0  \right) \right| \geq C_0  \log\{(k \vee q)n\}\{np(k+q)\}^{1/2}\left( n^{1/2} + p^{1/2}\right) \mid M_0 \right] \leq \frac{1}{n+p}.$$ Importantly, $C_0$ is an absolute constant not depending on the realization of $M_0$. This implies that $\text{pr}(A_4) \leq 1- 1/(n+p)$.
On the event $A_1 \cap A_2 \cap A_4$, we have  
\begin{equation*}\label{eq:D_cjmap_1}
    \begin{aligned}
        \mathcal{D}\left\{\Phi\left(Z_0\right) \mid  \Phi\left(\hat Z\right)\right\} \lesssim\log\{(k \vee q)n\}(k+q)^{1/2}\left( \frac{1}{n^{1/2}}+ \frac{1}{p^{1/2}}\right) +  \frac{4}{p}k\log(kn)+ \frac{2}{n} c_\lambda^2 k. 
    \end{aligned}
\end{equation*}
To conclude note that $\text{pr}[A_1 \cap A_2 \cap A_4] \geq 1 -\frac{2}{n} - \frac{C_1}{n+p} - \frac{1}{n+p}$ and recall $d_H^2(p , q) \leq \mathcal D(p \mid \mid q)$ and $d_{TV}^2(p , q) \leq \mathcal D(p \mid \mid q)$, where $d_H^2\{\Phi \left(\hat Z\right), \Phi \left(Z_0\right)\} = \frac{1}{np}\sum_{i=1}^n \sum_{j=1}^{p_n} d_H^2\{\Phi \left(\hat z_{ij}\right), \Phi \left(z_{0ij}\right)\}$, with  $d_H^2(f, g) = \left(\sqrt{f}-\sqrt{g}\right)^2 + \left(\sqrt{1-f}-\sqrt{1-g}\right)^2$, is the average Hellinger distance squared, and $d_{TV}^2\{\Phi \left(\hat Z\right), \Phi \left(Z\right)\} = \frac{1}\sum_{i=1}^n \sum_{j=1}^p\left|\Phi(\hat z_{ij}) - \Phi(z_{ij}) \right|^2$ is the average squared total variation distance.

\end{proof}

\begin{lemma}\label{lemma:sup_cjmap}
Consider the following set
\begin{equation*}
    \mathcal{G} =  \left\{Z \in \mathbb R^{n \times p}:~ \text{rank}(Z) = k+q, \quad ||Z||_{\infty} \leq \gamma_n  \right \}.
\end{equation*}
Then,
\begin{equation*}
    \text{pr} \left[\sup_{Z \in \mathcal G}\left| \mathcal L( Z)- E\left(\mathcal L( Z) \mid M_0  \right) \right| \geq C_0  \log\{(k \vee q)n\}\{np(k+q)\}^{1/2}\left( n^{1/2} + p^{1/2}\right) \mid M_0 \right] \leq \frac{1}{n+p},
\end{equation*}
where $C_0$ is an absolute constant.

\end{lemma}

\begin{proof}[Proof of Lemma \ref{lemma:sup_cjmap}]
The proof is similar to the one of Lemma 1 in \citet{davenport_1_bit}. We start by a straightforward application of the Markov inequality:
\begin{equation*}
\begin{aligned}
     &\text{pr} \left[\sup_{Z \in \mathcal G}\left| \mathcal L( Z)- E\left(\mathcal L( Z) \mid M_0  \right) \right| \geq C_0 \log\{(k \vee q)n\}\{np(k+q)\}^{1/2}\left( n^{1/2} + p^{1/2}\right) \mid M_0 \right] \\
     &= \text{pr} \left[\sup_{Z \in \mathcal G}\left| \mathcal L( Z)- E\left(\mathcal L( Z) \mid M_0  \right) \right|^h\geq \left[C_0  \log\{(k \vee q)n\}\{np(k+q)\}^{1/2}\left( n^{1/2} + p^{1/2}\right)\right]^h  \mid M_0\right]\\
     & \leq \frac{E \left(\sup_{Z \in \mathcal G}\left| \mathcal L( Z)- E\left(\mathcal L( Z) \mid M_0  \right) \right|^h \mid M_0\right)}{\left[C_0 \log\{(k \vee q)n\}\{np(k+q)\}^{1/2}\left( n^{1/2} + p^{1/2}\right)\right]^h }
\end{aligned}
\end{equation*}
Note that $ E \left(\sup_{Z \in \mathcal G}\left| \mathcal L( Z)- E\left(\mathcal L( Z) \mid M_0  \right) \right|^h \mid M_0 \right) =   E \left(\sup_{Z \in \mathcal G}\left| \mathcal{\bar L}( Z)- E\left(\mathcal {\bar L}( Z) \mid M_0  \right) \right|^h \mid M_0 \right)$, and, by a symmetrization argument, 
\begin{equation*}
  \begin{aligned}
        E &\left(\sup_{Z \in \mathcal G}\left| \mathcal{\bar L}( Z)- E\left(\mathcal {\bar L}( Z) \mid M_0  \right) \right|^h \mid M_0 \right)  \\
      & \leq 2^h   E \left(\sup_{Z \in \mathcal G}\left| \sum_{ij} \varepsilon_{ij}\left[1_{\{y_{ij} = 1\}}
      \log\left\{\frac{\Phi(Z_{ij})}{\Phi(0)}\right\}  - 1_{\{y_{ij} = 0\}}\log\left\{\frac{1-\Phi( Z_{ij})}{1-\Phi(0)}\right\} \right] \right|^h \mid M_0\right) 
    \end{aligned}  
\end{equation*}
where the $\varepsilon_{ij}$'s are independent Rademacher random variables and now the expectation is taken over $Y$ and also the $\varepsilon_{ij}$'s. For $|z| \in \gamma_n$, the functions $\frac{1}{L_{\gamma_n}}\log\left\{\frac{\Phi(z)}{\Phi(0)}\right\}$ and $\frac{1}{L_{\gamma_n}}\log\left\{\frac{1-\Phi(z)}{1-\Phi(0)}\right\}$ are contractions vanishing at $0$, where $L_{\gamma} =  \sup_{|z| \leq \gamma} \frac{|\Phi'(z)|}{\Phi(z) \{1- \Phi(z)\}} \leq 8(\gamma +1)$. Hence,
\begin{equation*}
\begin{aligned}
      E \left(\sup_{Z \in \mathcal G}\left| \mathcal L( Z)- E\left(\mathcal L( Z) \mid M_0  \right) \right|^h \mid M_0 \right) & \leq2^h \left(2L_{\gamma_n}\right)^h   E \left( \sup_{Z \in \mathcal G}\left|\sum_{i=1}^{n}\sum_{j=1}^{p} \varepsilon_{ij}\left(1_{\{y_{ij} = 1\}}  Z_{ij} - 1_{\{y_{ij} = 0\}} Z_{ij}\right)\right|^h\mid M_0 \right)\\
     & \leq2^h \left(2L_{\gamma_n}\right)^h   E\left( \sup_{Z \in \mathcal G}  \left| \langle E, Z \rangle\right|^h \mid M_0 \right) \\
     &=  \left(4 L_{\gamma_n}\right)^h   E\left( \sup_{Z \in \mathcal G}  \left| \langle E, Z \rangle\right|^h \right),
\end{aligned}
    \end{equation*}
    where $[E]_{ij} = \varepsilon_{ij}$.
Moreover, since $ \left| \langle A,  B \rangle\right| \leq ||A|| ||B||_{*}$ , 
    \begin{equation*}
    \begin{aligned}
      E\left( \sup_{Z \in \mathcal G}  \left| \langle E, Z \rangle\right|^h \right) &\leq  E\left( \sup_{Z \in \mathcal G}  ||E||^h ||Z||_{*}^h\right)\\
      &\leq \sup_{Z \in \mathcal G} ||Z||_* E\left(  ||E||^h \right)\\
      & \lesssim\left[ np(k+q) \log\{(k \vee q)n\}\right]^{h/2}\left(n^{\frac{h}{2}} + p^{\frac{h}{2}}\right)
    \end{aligned}
\end{equation*}
 where we used, \begin{equation*}
    E \left(||E||^h \right) \leq C\left(n^{\frac{h}{2}} + p^{\frac{h}{2}}\right)
\end{equation*}
by Theorem 1 of \citet{seginer_20},
and
\begin{equation*}
    ||Z||_{*}  \leq (np)^{1/2} \text{rank}^{1/2}(Z) ||Z||_{\infty} \leq \{np(k+q)\}^{1/2} \gamma_n \lesssim  [np (k+q) \log\{(k \vee q)n\}]^{1/2}
\end{equation*}
for sufficiently large $n$.
Thus, we get
\begin{equation*}
    \begin{aligned}
         E \left(\sup_{Z \in \mathcal G}\left| \mathcal L( Z)- E\left(\mathcal L( Z) \mid M_0  \right) \right|^h \mid M_0 \right)^{1/h} & \lesssim \log\{(k \vee q)n\}\{np(k+q)\}^{1/2}\left( n^{1/2} + p^{1/2}\right).
    \end{aligned}
\end{equation*}
Taking $C$ sufficiently large, and $h = \log(n+p)$ delivers
\begin{equation*}
    \begin{aligned}
        \frac{E \left(\sup_{Z \in \mathcal G}\left| \mathcal L( Z)- E\left(\mathcal L( Z) \mid M_0  \right) \right|^h \mid M_0\right)}{\left[C  \log\{(k \vee q)n\}\{np(k+q)\}^{1/2}\left( n^{1/2} + p^{1/2}\right)\right]^h } &\leq   \tilde C^{\log(n+p)} \leq \frac{1}{n+p}.
    \end{aligned}
\end{equation*}
\end{proof}

\section{Details on the Optimization Procedure}\label{subsec:optimization}
\subsection{Singular Value Decomposition based Initialization}\label{subsec:svd_init}
We initialize the optimization routine with the values obtained adapting the estimation procedure in \citep{chen_jmle} to the general design matrix $X$ case:
\begin{enumerate}
    \item We compute the singular value decomposition of $Y = U_{\tilde k} D_{\tilde k} V_{\tilde k}^\top + U_{-\tilde k} D_{-\tilde k} V_{-\tilde k}^\top$, where $\tilde k = k+q$, and let $\hat Y = U_{\tilde k} D_{\tilde k} V_{\tilde k}^\top $ be the $\tilde k$-rank approximation to $Y$.
    \item We compute $\tilde Y = [\tilde y_{ij}]_{ij}$ by applying the following thresholding operator entry-wise to $\hat Y = [\hat y_{ij}]_{ij}$, 
    \begin{equation*}
        \tilde y_{ij} = \begin{cases}
            \epsilon_{n,p} \quad &\text{if }  \hat y_{ij}< \epsilon_{n,p},\\
             \hat y_{ij} \quad &\text{if }  \hat y_{ij} \in [\epsilon_{n,p}, 1-\epsilon_{n,p}],\\
             1-\epsilon_{n,p} \quad &\text{otherwise.}  
        \end{cases}
    \end{equation*}
    \item We compute the matrix $\hat Z$ by applying the function $h^{-1}(\cdot)$ entry-wise to $\tilde Y$,
    \item We initialize $B$ via $\hat B =Z^\top X( X^\top  X)^{-1}$,
    \item We define $\hat{Z}^c$ as $\hat{Z}^c = \hat{Z} - X \hat B^\top$, 
     compute the singular value decomposition of $\hat{Z}^c = L_k S_k R_k^\top + L_{-k} S_{-k} R_{-k}^\top$, and initialize the factors via $\hat M = n^{1/2} L_k$ and loadings via $\hat \Lambda = \frac{1}{n^{1/2}} R_k S_k$.
\end{enumerate}

This method was initially proposed in \citet{chen_jmle} for the case where $X=1_n$, and \citet{ifa_svd} showed this produces consistent estimates for the loadings. Here, we report a heuristic argument:
\begin{enumerate}
    \item $\hat Y \approx  E(Y) = P = [pr( y_{ij}=1  \mid \eta_i, \lambda_{0j}, x_i, \beta_j)]_{ij} = [h \left(x_i^\top \beta_j +\eta_i^\top \lambda_{0j} \right)]_{ij}$  in a mean-squared sense as $n,p\to \infty$, hence $\hat y_{ij} \approx P_{ij} =h \left( x_i^\top \beta_j + \eta_i^\top \lambda_{0j} \right)$ on average across rows and columns.
    \item We ensure $\tilde y_{ij} \in [\epsilon_{n,p}, 1-\epsilon_{n,p}]$, so that we can treat $\tilde y_{ij}$ as a valid estimate of a probability and \say{invert} it in Step 3.
    \item From Step 1 and 2, we have $\tilde y_{ij} \approx h \left(x_i^\top \beta_j + \eta_i^\top \lambda_{0j} \right)$, hence, $ \hat z_{ij} = h^{-1} (\tilde Y_{ij})\approx  x_i^\top\beta_j + \eta_i^\top \lambda_{0j} $, thus, $\hat X = h^{-1} (\tilde Y) \approx X B_0^\top + M_{0} \Lambda_{0j}^\top$.
    \item Since $\hat X \approx  X B_0^\top+ M_0 \Lambda_0^\top$ and $(I-P_{\tilde X})\hat X =  L_k S_k R_k^\top + L_{-k} S_{-k} R_{-k}^\top$, then $\sqrt n L_k \approx M_0$ and $\frac{1}{n^{1/2}} R_k S_k \approx \Lambda_0$ up to orthogonal rotations,
    \item Similarly, since $\hat Z \approx XB_0^\top +  M_0\Lambda_0^\top$ and $E (X^\top M_0) = 0$, $\hat B^\top \approx (X^\top X)^{-1} X^\top (X B_0^\top + M_0\Lambda_0^\top) \approx B_0$.
\end{enumerate}
This choice for the initialization of the optimization routine guarantees that the initial values for $\left( M, \Lambda, B\right)$ are close to a local mode speeding up the convergence. 
When $n$ and $p$ are very large, we optionally replace the singular value decomposition with a randomized version \citep{rsvd1, rsvd2}. This considerably speeds-up  initialization with negligible impact on the final solution of the algorithm.

\subsection{Projected Newton-Raphson Ascent}\label{subsec:newton}
Each subproblem of \eqref{eq:optimization_params} and \eqref{eq:optimization_factors} is solved via projected Newton ascent. More specifically, for \eqref{eq:optimization_params}, we maximize the log-posterior for $\theta_j = (\beta_j, \lambda_j)$ for $j=1, \dots, p$, given the previous estimate for $M$, $\hat M$. This is equivalent to finding
maximum a posterior estimate 
for the regression coefficients of $p$ logistic regressions using the concatenation of $X$ and $\hat M$ and the $Y^{(j)}$'s as the outcome, where $Y^{(j)}$ denotes the $j$-th column of $Y$. 
This is solved via projected Newton ascent, that is, given the current value of the parameter $\theta_j^{(t)}$
we update via
\begin{equation}\label{eq:theta_j_t_plus_1}
    \theta_j^{(t+1)*} \gets \theta_{j}^{(t)} - \nu \nabla^2_{\theta_j \theta_j^\top} \log p(\hat M, \hat \Lambda, \hat B \mid Y) \mid_{\theta_j = \theta_j^{(t)}} ^{-1} \nabla_{\theta_j} \log p(\hat M, \hat \Lambda, \hat B \mid Y) \mid_{\theta_j = \theta_j^{(t)}},
\end{equation}
and
\begin{equation}
    \theta_j^{(t+1)} \gets T_1(\theta_j^{(t+1)*} ),
\end{equation}
where $T_1(x)$ projects $x$ to the constraint set for the $\theta_j$'s defined by the problem. We set the step-size to the default value of $\nu=0.3$, which worked well in both simulated and real data examples. 

Similarly, for \eqref{eq:optimization_factors}, we maximize the log-posterior over $\eta_i$, given the estimates for $\Lambda$ and $B$, $\hat \Lambda$, $\hat B$ for $i=1, \dots, n$. Similarly as above, this operation can be parallelized across rows and at each iteration let 
\begin{equation}\label{eq:eta_i_t_plus_1}
    \eta_i^{(t+1)*} \gets \eta_i^{(t)} - \nu \nabla^2_{\eta_i \eta_i^\top} \log p(\hat M, \hat \Lambda, \hat B \mid Y) \mid_{\eta_i = \eta_i^{(t)}} ^{-1} \nabla_{\eta_i} \log p(\hat M, \hat \Lambda, \hat B \mid Y) \mid_{\eta_i = \eta_i^{(t)}},
\end{equation}
and
\begin{equation}
    \eta_i^{(t+1)} \gets T_2(\eta_i^{(t+1)*} )
\end{equation}
where $T_2(x)$ projects $x$ to the constraint set for the $\eta_i$'s defined by the problem. We set the step-size to the default value of $\nu=1$, which worked well in both simulated and real data examples. 

We stop the algorithm when the Euclidean norm of the update is smaller than some small threshold, which we set to the default value of $0.001$. To optimize \eqref{eq:log_joint_posterior}, we alternate between \eqref{eq:optimization_params} and  \eqref{eq:optimization_factors} until the relative increase in the log-posterior is smaller than a small threshold, which we set by default to $0.001$.

Our current implementation uses \texttt{for loops} to iterate steps in \eqref{eq:theta_j_t_plus_1} and \ref{eq:eta_i_t_plus_1}, but these steps could be parallelized across columns ($j=1, \dots, p$) and rows ($i=1, \dots, n$) of $Y$ to produce substantial gains in computational speed. Not all of the competitors are similarly parallelizable.

\subsection{Post-Processing Procedure}\label{subsec:postprocessing}
We borrow the post-processing procedure from \citet{chen_jmle} and adapt it to the case of a general design matrix $X$. In particular, we transform the $\big(\hat M, \hat \Lambda, \hat B\big)$ solution from \eqref{eq:cjmap} to $\left(\tilde M, \tilde \Lambda, \tilde B \right)$, through the following operations.
\begin{enumerate}
    \item Compute $\hat M^c = \hat M - P_X \hat M$, where $P_{A} = A(A^\top A)^{-1}A^\top$, and apply the singular value decomposition to $\hat M^c = U D V^\top$,
    \item Set $\tilde M = n^{1/2} U$ and $\tilde \Lambda = \frac{1}{n^{1/2}}\hat \Lambda V $,
    \item Set $\tilde B = \hat B + \hat \Lambda \hat M^\top X ( X^\top  X)^{-1}$.
\end{enumerate}
It is easy to verify that $\left( \tilde M, \tilde \Lambda, \tilde B\right)$ satisfies the following properties
\begin{enumerate}
    \item $\hat M\hat \Lambda^\top + X \hat B^\top = \tilde M\tilde \Lambda^\top + X\tilde B^\top$,
    \item $\tilde M^\top \tilde M = nI_k$,
    \item $\tilde M^\top X = 0$.
\end{enumerate}

\section{Hyperparameter Selection}\label{subsec:hyperparams}
\subsection{Variance Parameters}
We highlight a data-driven strategy to select the hyperparameters $\tau_B =(\tau_{\beta_1}, \dots, \tau_{\beta_p})$ and $\tau_\Lambda =(\tau_{\lambda_1}, \dots, \tau_{\lambda_p})$. 
In particular, given the initial values of $(\hat \Lambda, \hat B)$ obtained with the procedure described in Section \ref{subsec:svd_init}, we set $\tau_{\lambda_j}$ to $\hat \tau_{\lambda_j} = \mathcal{T}\left( k^{-1/2}\left|\left|\hat \lambda_j\right|\right|\right)$ and $\tau_{\beta_j}$ to $\hat \tau_{\beta_j} = \mathcal{T}\left(  k^{-1/2}\left|\left|\hat \beta_j\right|\right|\right)$, where $\mathcal{T}$ is a hard-truncation operator such that $\mathcal{T}(x) = x 1_{\{x \in (l,u) \}} + l1_{\{x \leq l \}} + u 1_{\{x \geq u \}}$ and $l = 0.5, u=20$. 
This strategy is appealing since it does not require any manual tuning and showed good performance across a wide range of examples. 
Clearly, other choices including posterior predictive checks or cross-validation maximizing the likelihood on a hold out set are equally viable but can be computationally more expensive. .
\subsection{Calibration of $\rho$}\label{subsec:rho}
We calibrate the variance inflation factor $\rho$ by choosing $\rho = \max_{1 \leq j,j' \leq p} b_{jj'}$ with 
\begin{align}\label{eq:b_jj}
\begin{split}
 b_{jj'} = \begin{cases}
        \big\{1 + \frac{||\tilde \lambda_j||_2^2||\tilde \lambda_{j'}||_2^2 + \left(\tilde \lambda_j^\top \tilde \lambda_{j'}\right)^2}{\tilde \sigma_{j'}^2||\tilde \lambda_j||_2^2 + \tilde \sigma_j^2||\tilde \lambda_{j'}||_2^2} \big\}^{1/2} \quad & \text{if } j \neq j'\\
         \big(1 + \frac{||\tilde\lambda_j||_2^2}{2\tilde \sigma_j^2} \big)^{1/2} \quad & \text{otherwise}
         \end{cases}
\end{split}
\end{align}
where $\tilde \sigma_j^2 = 1.702^2 + \frac{n}{\sum_{i=1}^nh \left(x_i^\top \tilde\beta_j + \tilde\eta_i^\top \tilde \lambda_j \right)\left\{1-h \left(x_i^\top \tilde\beta_j + \tilde \eta_i^\top \tilde\lambda_j \right)\right\}}$.
This choice is inspired by 
\citet{fable}. Here, we provide a heuristic justification. 
Considering $Y$ as the dichotomization of a latent continuous matrix $Y^*$, we have
    $y_{ij} = 1_{\{y_{ij}^* >0\}}$, $i=1, \dots, n, j=1, \dots, p$,
where the matrix $Y^* = [y_{ij}^*]_{ij}$ is 
\begin{equation*}
    Y^* = X B^\top + M\Lambda^\top + E, \quad E = [\epsilon_{ij}]_{ij}, \quad \epsilon_{ij} \sim f,  
\end{equation*}
and $f$ is a logistic density, that is $f(x) = h'(x)$ with $h(\cdot)$ being the logistic cumulative density function. Since $\sup_{x \in \mathbb R}\left|h(x) - \Phi(x/1.702)\right| \leq 0<0.0095$ \citep{logistic_normal_approx_09}, we approximate the $\epsilon_{ij}$'s via a normally distributed random variable as  $\epsilon_{ij} \approx 1.702 \varepsilon_{ij}$ where $\varepsilon_{ij}\sim  N(0,1)$ independently. Defining $\tilde E = \tilde Z  - Y^*$, we have
\begin{equation*}
    \tilde Z \approx X B^\top + M\Lambda^\top  + 1.702 E + \tilde E, \quad E = [\varepsilon_{ij}]_{ij},\quad \varepsilon_{ij} \sim N(0,1).
\end{equation*}
\citet{fable} consider the case of Gaussian distributed data without covariates and estimate the latent factors as the leading left singular vectors of the data scaled by $n^{1/2}$. The authors showed that defining the coefficients $\{b_{jj'}\}$ as \eqref{eq:b_jj}, where $\tilde \lambda_j$ is the posterior mean of $\lambda_j$, $\tilde \sigma_j$ estimates the standard deviation of the columns of the residual matrix, and choosing $\rho = \max_{j,j'} b_{jj'}$ guarantees asymptotic correct frequentist coverage. 
Thus, considering the analogy developed above, since $\tilde M$ corresponds to the left singular values of $\tilde Z$, after regressing out the covariate effects, scaled by $n^{1/2}$ (up to rotation), we apply the coverage correction strategy from \citet{fable}, treating the $\tilde \lambda_j$'s as if they were the regression coefficient from regressing $\tilde M$ on $\tilde Z$.  
As a final step, we approximate the elements $\tilde E$ as independent normal random variables. Consider the negative Hessian of the log-likelihood for the $i$-th row and $j$-th column with respect to $\theta_j$
\begin{equation*}
    -\nabla^2_{\theta_j} \log p(y_{ij} \mid -) =  h \left(\tilde z_{ij}\right)\left\{1-h\left(\tilde z_{ij}\right)\right\} \tilde x_i \tilde x_i^\top,
\end{equation*}
where $\tilde x_i = (x_i^\top ~ \tilde \eta_i^\top)^\top$ and $\tilde z_{ij} = x_i^\top \tilde \beta_j + \tilde \eta_i^\top \tilde \lambda_j$. Thus, we can interpret $ h \left(\tilde z_{ij}\right)\left\{1-h \left(\tilde z_{ij}\right)\right\}$ as a proxy for the precision for the estimator from the $i$-th observation for the $j$-th outcome. Hence, we estimate the variance of the $j$-th column of $\tilde E$ as the inverse of the average of the precision proxies, that is $\frac{n}{\sum_{i=1}^nh \left(\tilde z_{ij}\right)\left\{1-h \left(\tilde z_{ij}\right)\right\}}$,  and set
$\tilde \sigma_j^2 = 1.702^2 + \frac{n}{\sum_{i=1}^nh \left(x_i^\top \tilde\beta_j + \tilde\eta_i^\top \tilde \lambda_j \right)\left\{1-h \left(x_i^\top \tilde\beta_j + \tilde \eta_i^\top \tilde\lambda_j \right)\right\}}$. 

Our derivation for $\rho$ is based on the representation of binary variables as truncated latent continuous variables with logistic density. Similarly, an analogous approach could be derived for Bernoulli responses with the probit link, whereas extensions to other GLLVMs are less straightforward. 

\section{Additional Experiments}\label{sec:additional_experiments}
\subsection{Lower Dimensional Scenarios}

\begin{table}[H]
\centering
{	\begin{tabular}{crrrcrrr}
		&  \multicolumn{7}{c}{$p=50$}  \\
& \multicolumn{3}{c}{$n=100$}& &\multicolumn{3}{c}{$n=500$} \\
		Method & $\Lambda \Lambda^\top$ & $B$ & time (s) & & $\Lambda \Lambda^\top$ & $B$ & time (s)  \\
		\texttt{GMF - Newton} &$>100$ & $39.62^{2.29}$&  $0.31^{0.09}$ & &$61.16^{8.71}$& $20.40^{8.33}$& $0.79^{0.01}$\\ 
  		\texttt{GMF - Airwls} &$>100$  & $>100$ &$6.21^{1.72}$ & &$>100$ &$>100$& $13.69^{2.79}$\\
    \texttt{GLLVM - LA}  &$>100$  & $44.48^{2.02}$ & $13.47^{1.12}$ & & $65.38^{0.29}$ & $40.86^{0.17}$& $74.14^{2.49}$ \\
        \texttt{GLLVM - EVA}  &$>100$  & $54.38^{4.55}$ & $4.97^{0.21}$ & & $27.07^{0.54}$& $14.76^{0.23}$& $61.05^{1.11}$\\
		\texttt{FLAIR} & $70.92^{1.70}$& $31.74^{0.49}$& $0.08^{0.01}$ & & $48.69^{1.38}$ & $16.44^{0.27}$ & $0.36^{0.01}$ \\
&  \multicolumn{7}{c}{$p=100$}  \\
& \multicolumn{3}{c}{$n=100$}& &\multicolumn{3}{c}{$n=500$} \\
 Method & $\Lambda \Lambda^\top$ & $B$ & time (s) & & $\Lambda \Lambda^\top$ & $B$ & time (s) \\
		\texttt{GMF - Newton} &$75.59^{7.30}$ & $43.20^{1.38}$ & $0.38^{0.07}$ & &$43.53^{6.03}$ & $22.20^{0.89}$ & $0.95^{0.05}$ \\ 
  		\texttt{GMF - Airwls} & $>100$ & $>100$ & $13.67^{2.64}$&  & $>100$ & $>100$ & $103.55^{10.55}$ \\
        \texttt{GLLVM - LA}  & $80.18^{15.82}$ & $43.89^{1.59}$ & $44.20^{2.81}$& & $65.70^{0.31}$& $42.26^{0.11}$& $250.44^{17.52}$\\
                \texttt{GLLVM - EVA}  & $87.82^{7.83}$  & $39.96^{0.82}$ & $21.12^{1.09}$ & & $25.11^{0.56}$& $14.88^{0.18}$& $128.93^{7.52}$\\
		\texttt{FLAIR} & $53.55^{0.81}$ & $32.01^{0.41}$ & $0.22^{0.01}$ & & $27.23^{0.71}$ & $15.09^{0.19}$ & $0.56^{0.01}$ \\
  &  \multicolumn{7}{c}{$p=200$}  \\
& \multicolumn{3}{c}{$n=100$}& &\multicolumn{3}{c}{$n=500$} \\
 Method & $\Lambda \Lambda^\top$ & $B$ & time (s) & & $\Lambda \Lambda^\top$ & $B$ & time (s) \\
		\texttt{GMF - Newton} &$53.65^{2.25}$ & $33.61^{0.66}$& $1.96^{0.55}$ & & $34.01^{6.04}$ &$16.81^{0.40}$& $2.09^{0.10}$\\ 
  		\texttt{GMF - Airwls} &$>100$  & $>100$ & $16.62^{5.37}$& &$31.15^{5.39}$ &$18.36^{0.72}$& $170.25^{27.89}$\\
        \texttt{GLLVM - LA}  & $79.07^{13.96}$ & $46.69^{3.21}$ & $132.78^{8.89}$ & & $66.08^{0.26}$& $42.18^{0.07}$ & $848.70^{106.34}$ \\
        \texttt{GLLVM - EVA}  & $73.29^{2.67}$ & $40.54^{0.83}$ & $26.40^{1.32}$& & $23.80^{0.35}$ & $15.05^{0.16}$& $267.89^{18.46}$ \\
		\texttt{FLAIR} & $50.49^{0.58}$ & $32.02^{0.33}$& $0.75^{0.14}$ & & $22.21^{0.21}$ & $14.69^{0.15}$& $0.94^{0.01}$ 
	\end{tabular}}
	\label{tablelabel}
	\caption{Comparison of the methods in terms of estimation accuracy. Root normalized squared error for $\Lambda \Lambda^\top$ and $B$, and running time.  We report mean and standard error over 50 replications. Estimation errors have been multiplied by $10^2$. 
 \texttt{GMF - Newton} and \texttt{GMF - Airwls} denote \citet{gmf}'s method fitted via the quasi Newton method and via alternating iteratively reweighted least square algorithm respectively. 
\texttt{GLLVM - LA} and \texttt{GLLVM - EVA} denote the generalized linear latent variable model fitted via the Laplace approximation and extended variational approximation respectively.
	}
	\label{tab:accuracy_low_dimensional}
\end{table}

We conducted a simulation study to test the performance of \texttt{FLAIR} in a lower dimensional scenario. In particular, we simulate data from model \eqref{eq:gllvm}, where parameters are generated as follows
\begin{equation*}\begin{aligned}
\lambda_{0jl} & \sim TN(0, \sigma^2, -5, 5), \quad \beta_{0jl'}\sim TN(0,\sigma^2, -5, 5)\\ %
\end{aligned}  \end{equation*} for $j=1, \dots, p$, $l=1, \dots, k$, $l'=1, \dots, q$, and $TN(\mu, \sigma^2, a, b)$ denotes a truncated normal distribution with mean $\mu$, variance $\sigma^2$, and support $(a, b)$. We let the sample and outcome sizes be $(n, p) \in \{100, 500\} \times \{50, 100, 200\}$, and we set $k=q=2$ and $\sigma^2=1$,
For each configuration, we replicate the experiment 50 times. 
We compare our model to \texttt{GMF} and to a generalized linear latent variable model fitted by the Laplace approximation (\cite{gllvm_va}, \texttt{GLLVM-LA}) and an extended variational approximation (\cite{gllvm_eva}, \texttt{GLLVM-EVA}) as implemented in the \texttt{gllvm} R package. The \texttt{gllvm} package provides an estimate of the covariance of the model estimates. For $\Lambda \Lambda^\top$, we obtained Monte Carlo estimates of the confidence intervals. For each method, we set the tuning parameters equal to their default values. The number of latent factors was estimated using the information criterion discussed in Section \ref{subsec:choice_k}, which always picked the correct value.
Table \ref{tab:accuracy_low_dimensional} reports a comparison in terms of estimation accuracy and computational time. Even in lower dimensional examples \texttt{FLAIR} compares favorably to competitors in many scenarios. Table \ref{tab:uq_additional} provides additional evidence of the frequentist validity of \texttt{FLAIR} credible intervals; these intervals had valid frequentist coverage on average over entries of $B$ and $\Lambda \Lambda^\top$ for $p \geq 100$ while suffering only a mild under-coverage for $p=50$. In contrast, \texttt{GLLVM-LA} suffers from undercoverage, and \texttt{GLLVM-EVA} provides valid uncertainty quantification for $B$ but not for $\Lambda \Lambda^\top$.

\begin{table}
\centering
{   
    \begin{tabular}{crrcrr}
   & \multicolumn{5}{c}{$p=50$} \\
    & \multicolumn{2}{c}{$n=100$} & & \multicolumn{2}{c}{$n=500$} \\
    Method & $\Lambda \Lambda^\top$ & $B$ & & $\Lambda \Lambda^\top$ & $B$\\ 
    \texttt{GLLVM - LA}& $76.48^{4.92}$ & $78.54^{0.36}$ & & $66.60^{4.46}$ & $58.68^{3.07}$ \\ 
       \texttt{GLLVM - EVA} & $79.52^{4.54}$ & $97.27^{0.30}$ & & $67.18^{4.72}$& $97.60^{0.29}$ \\ 
    \texttt{FLAIR} & $96.43^{0.32}$ & $92.14^{0.46}$ & & $93.60^{0.28}$ & $89.90^{0.53}$ \\
    vanilla \texttt{FLAIR} ($\rho=1$) & $89.96^{0.54}$ & $84.22^{0.61}$ & & $86.16^{0.39}$ & $81.32^{0.69}$ \\
   & \multicolumn{5}{c}{$p=100$} \\
        & \multicolumn{2}{c}{$n=100$} & & \multicolumn{2}{c}{$n=500$} \\
    Method & $\Lambda \Lambda^\top$ & $B$ & & $\Lambda \Lambda^\top$ & $B$\\ 
    \texttt{GLLVM - LA} & $67.39^{5.75}$ & $79.18^{0.35}$ & & $86.20^{2.43}$ & $56.62^{0.25}$ \\ 
    \texttt{GLLVM - EVA}&  $71.55^{5.50}$ & $97.50^{0.24}$ & & $90.67^{2.63}$  & $97.74^{0.19}$  \\ 
    \texttt{FLAIR} & $97.59^{0.19}$ & $97.20^{0.23}$ & & $97.71^{0.12}$ & $96.69^{0.24}$ \\
    vanilla \texttt{FLAIR} ($\rho=1$) & $92.13^{0.31}$ & $93.20^{0.39}$ & & $92.03^{0.22}$ & $92.62^{0.40}$ \\
    &\multicolumn{5}{c}{$p=200$} \\
        & \multicolumn{2}{c}{$n=100$} & & \multicolumn{2}{c}{$n=500$} \\
    Method & $\Lambda \Lambda^\top$ & $B$ & & $\Lambda \Lambda^\top$ & $B$\\ 
    \texttt{GLLVM - LA}& $70.48^{3.07}$ & $79.16^{0.21}$ & & $34.37^{3.23}$  &  $56.08^{0.22}$\\ 
    \texttt{GLLVM - EVA} & $78.00^{3.62}$ & $97.63^{0.18}$ & & $39.47^{3.60}$  & $97.70^{0.56}$ \\ 
\texttt{FLAIR} & $98.44^{0.14}$ & $96.41^{0.22}$ & & $97.70^{0.15}$ & $96.15^{0.22}$ \\
    vanilla \texttt{FLAIR} ($\rho=1$) & $93.15^{0.27}$ & $90.31^{0.40}$ & & $92.44^{0.26}$ & $90.59^{0.38}$ \\
    \end{tabular}}
    \label{tablelabel}
    \caption{Comparison of the methods in terms of uncertainty quantification. Average frequentist coverage for entries of $\Lambda \Lambda^\top$ and $B$. We report mean and standard deviation over 50 replications. Coverage values have been multiplied by $10^2$.
    We report mean and standard error over 50 replications.  Coverage values have been multiplied by $10^2$. \texttt{GLLVM - LA} and \texttt{GLLVM - EVA} denote the generalized linear latent variable model fitted via the Laplace approximation and extended variational approximation respectively. For \texttt{FLAIR}, we report results with and without applying the correction factor $\rho$ to the posterior variance.}
    \label{tab:uq_additional}
\end{table}

\subsection{Longitudinal Scenarios}\label{subsec:longitudinal_experiments}
We present some experiments in longitudinal scenarios.
In particular, we compare with \citet{lee24} (\texttt{LVHML}, henforth), implemented using the code  at \href{https://github.com/Arthurlee51/LVHML}{https://github.com/Arthurlee51/LVHML}. As discussed in the introduction, \texttt{LVHML} assumes latent factors to be fixed constants. Therefore, we cannot directly estimate the latent covariance between outcomes with $\hat \Lambda \hat \Lambda^\top$. 
In fact, the latent covariance between species implied by model \eqref{eq:gllvm} is given by $\Lambda cov(\eta_i) \Lambda^\top$, which reduces to $\Lambda \Lambda^\top$ in \eqref{eq:model_methodology_integrated}, since $cov(\eta_i) = I_k$ as we assumed $\eta_i \sim N_k(0, I_k)$. To remedy this and estimate the latent covariance $\Lambda \Lambda^\top$ in \eqref{eq:model_methodology_integrated} for \texttt{LVHML}, we use $\hat \Lambda \big(\hat M^\top \hat M) /n \hat \Lambda^\top $, which is obtained by replacing $cov(\eta_i)$ by the empirical covariance of estimates for latent factors, where $\hat M$ and $\hat \Lambda$ are estimates for latent factors and factor loadings from \texttt{LVHML}.

Even though \texttt{LVHML} can be seen as a more general version of the model presented in \texttt{FLAIR}, its current implementation does not allow model fitting if only one time point is used (that is when $T=1$), which prevents a comparison when data is generated according to model \eqref{eq:gllvm}. Instead, we consider two longitudinal scenarios, both with $T=2$: in the first one (scenario (a)), for each $i,j$, we observe the outcome at both time points, in the second one (scenario (b)), we observe all the outcomes at the first time point, and only outcomes from the first sample at the second time point. We generate the parameters as in Section \ref{sec:simulation} with the same values of $\sigma$, $k$, and $q$, and $(n, p) \in \{500, 1000\} \times \{1000, 10000\}$.

For \texttt{FLAIR}, we neglect the longitudinal structure and fit the methodology concatenating data from different time points and considering samples of the same unit $i$ from two time points as different (independent) samples. In comparing the estimation accuracy of the regression coefficients, we excluded the intercept, which is not comparable between the two models.
For uncertainty quantification of \texttt{LVHML} estimates, we construct confidence intervals for the $\beta_j$'s using Theorem 3 in \citet{lee24}. For the $\lambda_j$'s, we sample $n_{MC}$ samples from a normal distribution with variances implied by Theorem 6 in \citet{lee24}, obtain the corresponding samples for  $ \Lambda \big(\hat M^\top \hat M) /n  \Lambda^\top $, and estimate Monte Carlo confidence intervals.

\begin{table}[H]
\centering
{	\begin{tabular}{crrrcrrr}
		&  \multicolumn{7}{c}{$p=1000$}  \\
            		&  \multicolumn{7}{c}{Scenario (a)}  \\
& \multicolumn{3}{c}{$n=500$}& &\multicolumn{3}{c}{$n=1000$} \\
		Method & $\Lambda \Lambda^\top$ & $B$ & time (s) & & $\Lambda \Lambda^\top$ & $B$ & time (s)  \\
				\texttt{LVHML} & $ 33.92^{0.23}$ & $ 12.26^{0.06}$  & $35.25^{2.32}$ & &$ 23.11^{0.14}$ & $ 8.38^{0.03}$ &$73.51^{4.62}$  \\
		\texttt{FLAIR} & $ 29.32^{0.12}$ & $11.48^{0.05}$  & $ 25.28^{1.41}$ & &$ 20.77^{0.08}$ & $ 8.10^{0.03}$ &$ 56.87^{3.91}$  \\
    &  \multicolumn{7}{c}{Scenario (b)}  \\
& \multicolumn{3}{c}{$n=500$}& &\multicolumn{3}{c}{$n=1000$} \\
		Method & $\Lambda \Lambda^\top$ & $B$ & time (s) & & $\Lambda \Lambda^\top$ & $B$ & time (s)  \\
    \texttt{LVHML} &  $ 54.22^{0.32}$  &$16.72^{0.05}$ & $15.28^{1.43}$ & &$35.00^{0.19}$ & $ 11.09^{0.03}$ &$45.77^{4.01}$  \\
		\texttt{FLAIR} & $38.82^{0.12}$ & $ 14.41^{0.05}$ & $10.24^{0.56}$ & &$ 27.25^{0.07}$ & $ 10.19^{0.03}$ &$20.83^{1.98}$  \\
&  \multicolumn{7}{c}{$p=10000$}  \\
            		&  \multicolumn{7}{c}{Scenario (a)}  \\
& \multicolumn{3}{c}{$n=500$}& &\multicolumn{3}{c}{$n=1000$} \\
		Method & $\Lambda \Lambda^\top$ & $B$ & time (s) & & $\Lambda \Lambda^\top$ & $B$ & time (s)  \\
				\texttt{LVHML} &  $ 33.53^{0.29}$ & $11.73^{0.06}$ & $2743.91^{78.22}$ & &$21.68^{0.01}$ & $ 8.21^{0.01}$  &$4756.29^{92.19}$  \\
		\texttt{FLAIR} &  $ 29.36^{0.17}$ & $ 11.38^{0.05}$ & $180.53^{13.39}$ & &$20.70^{0.08}$ & $ 8.02^{0.03}$ &$ 432.56^{27.21}$  \\
    &  \multicolumn{7}{c}{Scenario (b)}  \\
& \multicolumn{3}{c}{$n=500$}& &\multicolumn{3}{c}{$n=1000$} \\
		Method & $\Lambda \Lambda^\top$ & $B$ & time (s) & & $\Lambda \Lambda^\top$ & $B$ & time (s)  \\
				\texttt{LVHML} & $ 46.85^{0.23}$ & $ 16.04^{0.05}$ & $584.26^{37.68}$ &  & $30.00^{0.01}$ & $10.70^{0.03}$ & $ 754.70^{37.30}$  \\
		\texttt{FLAIR}&$39.51^{0.12}$ & $14.29^{0.04}$ &$121.15^{3.61}$ & &$ 27.44^{0.07}$ & $10.10^{0.03}$ &$154.46^{10.75}$  \\
	\end{tabular}}
	\label{tablelabel}
\caption{Comparison of the methods in terms of estimation accuracy in the longitudinal simulation experiments. Root normalized squared error for $\Lambda \Lambda^\top$ and $B$, and running time.  We report mean and standard deviation over 50 replications.  Estimation errors have been multiplied by $10^2$. \texttt{LVHML} denote \citet{lee24}'s method. For one dataset with $n=1000$, $p=1000$, and two datasets with $n=1000$ and $p=10000$, both in scenario (b), \texttt{LVHML} ran into numerical error. }
	\label{tab:accuracy_longitudinal}
\end{table}

Table \ref{tab:accuracy_longitudinal} reports the relative estimation error for $\Lambda \Lambda^\top$ and $B$ and the running times. In all scenarios, \texttt{FLAIR} offers better estimation accuracy and shorter running times, even by a factor $>10$ in high-dimensional scenarios. Table \ref{tab:accuracy_beta_longitudinal} shows the median and maximum root mean squared error in estimating individual $\beta_j$'s (the rows of $B$). For both metrics, \texttt{FLAIR} outperforms the competitor with notable decreases in
maximum error. Table \ref{tab:uq_longitudinal} reports the average coverage over entries of $B$ and $\Lambda \Lambda^\top$.  
\texttt{LVHML} intervals suffer from under-coverage in all cases. In contrast, \texttt{FLAIR} obtains valid average coverage in all experiments, except for a mild under-coverage of $B$ in the Scenario (a), that is when \texttt{FLAIR} is more misspecified. \texttt{FLAIR} without the coverage correction (that is when $\rho=1$) has higher coverage than \texttt{LVHML} with intervals that are shorter or of comparable length. 

 \texttt{LVHML} showed numerical instabilities, encountering numerical errors for one dataset with $n=1000$, $p=1000$, and two datasets with $n=1000$ and $p=10000$, both in scenario (b).

\begin{table}
\centering
{	\begin{tabular}{crrcrr}
		&  \multicolumn{5}{c}{$p=1000$}  \\
            		&  \multicolumn{5}{c}{Scenario (a)}  \\
& \multicolumn{2}{c}{$n=500$}& &\multicolumn{2}{c}{$n=1000$} \\
		Method & Median & Max & & Median & Max\\
				\texttt{LVHML} & $ 11.12^{0.05}$&$ 30.72^{0.53}$  &  &$ 7.66^{0.03}$ &$ 20.39^{0.33}$ \\
		\texttt{FLAIR} & $ 10.62^{0.05}$&$ 25.48^{0.30}$  &  &$ 7.48^{0.02}$ &$18.36^{0.24}$ \\
         		&  \multicolumn{5}{c}{Scenario (b)}  \\
& \multicolumn{2}{c}{$n=500$}& &\multicolumn{2}{c}{$n=1000$} \\
		Method & Median & Max & & Median & Max\\
				\texttt{LVHML} & $ 14.93^{0.05}$&$ 50.16^{1.22}$  &  &$ 10.10^{0.03}$ &$ 28.84^{0.52}$ \\
		\texttt{FLAIR} & $ 13.35^{0.04}$&$ 30.42^{0.38}$  &  &$9.53^{0.02}$ &$21.50^{0.26}$ \\        		
        &  \multicolumn{5}{c}{$p=10000$}  \\
        &  \multicolumn{5}{c}{Scenario (a)}  \\
        & \multicolumn{2}{c}{$n=500$}& &\multicolumn{2}{c}{$n=1000$} \\
		Method & Median & Max & & Median & Max\\
				\texttt{LVHML} & $10.76^{0.05}$&$ 33.89^{0.51}$  &  & $ 7.55^{0.03}$& $ 22.77^{0.28}$ \\
		\texttt{FLAIR} & $ 10.55^{0.05}$&$29.11^{0.34}$  &  &$ 7.42^{0.03}$ &$ 20.25^{0.21}$ \\
         		&  \multicolumn{5}{c}{Scenario (b)}  \\
& \multicolumn{2}{c}{$n=500$}& &\multicolumn{2}{c}{$n=1000$} \\
		Method & Median & Max & & Median & Max\\
				\texttt{LVHML} & $ 14.51^{0.04}$&$ 54.90^{0.11}$  &  &$ 9.85^{0.02}$ &$ 32.00^{0.49}$ \\
		\texttt{FLAIR} & $ 13.41^{0.04}$&$ 33.89^{0.31}$  &  &$ 9.47^{0.02}$ &$ 24.53^{0.27}$ \\ 
	\end{tabular}}
	\label{tablelabel}
	\caption{Comparison of the methods in terms of estimation accuracy for the individual $\beta_{j}$'s in the longitudinal experiments. Median and Maximum root mean squared error for the individual $\beta_{j}$'s. We report mean and standard deviation over 50 replications. Estimation errors have been multiplied by $10^2$. \texttt{LVHML} denote \citet{lee24}'s method.}
	\label{tab:accuracy_beta_longitudinal}
\end{table}

\begin{table}[h]
\centering
{   
    \begin{tabular}{crrrrcrrrr}
    &  \multicolumn{9}{c}{$p=1000$}  \\
            		&  \multicolumn{9}{c}{Scenario (a)}  \\
                    & \multicolumn{4}{c}{$n=500$}& &\multicolumn{4}{c}{$n=1000$}\\
& \multicolumn{2}{c}{Coverage}& \multicolumn{2}{c}{Length} & & \multicolumn{2}{c}{Coverage}& \multicolumn{2}{c}{Length} \\
    Method & $\Lambda \Lambda^\top$ & $B$ &  $\Lambda \Lambda^\top$ & $B$ & & $\Lambda \Lambda^\top$ & $B$ &  $\Lambda \Lambda^\top$ & $B$\\ 
    \texttt{LVHML}& $89.95^{0.21}$ &$86.80^{0.17}$  &$ 0.88$ &$0.36 $ & &$89.56^{0.20}$ & $87.02^{0.15}$ &$0.60$ & $0.25$\\ %
    \texttt{FLAIR} &$95.35^{0.12}$ &$93.15^{0.14}$  &$ 0.96$  & $0.43 $ & &$95.10^{0.14}$ & $93.23^{0.11}$ &$0.67$ & $0.30$ \\
    vanilla \texttt{FLAIR} ($\rho=1$) & $90.68^{0.17}$ &$88.12^{0.16}$  &$0.80 $ & $ 0.36$  & &$90.15^{0.19}$ & $88.01^{0.14}$&$0.56$ & $0.25$\\
    &  \multicolumn{9}{c}{Scenario (b)}  \\
                    & \multicolumn{4}{c}{$n=500$}& &\multicolumn{4}{c}{$n=1000$}\\
& \multicolumn{2}{c}{Coverage}&\multicolumn{2}{c}{Length} & & \multicolumn{2}{c}{Coverage} &\multicolumn{2}{c}{Length} \\
    Method & $\Lambda \Lambda^\top$ & $B$  & $\Lambda \Lambda^\top$ & $B$ & & $\Lambda \Lambda^\top$ & $B$  & $\Lambda \Lambda^\top$ & $B$\\ 
    \texttt{LVHML}& $90.16^{0.18}$ &$89.71^{0.11}$  &$1.39$ & $0.52$& &$89.90^{0.19}$ & $90.24^{0.09}$ &$0.90$ & $0.35$\\ %
    \texttt{FLAIR} &$96.88^{0.10}$ &$95.75^{0.08}$  &$1.39$ & $0.61$& &$96.66^{0.09}$ & $95.75^{0.07}$& $0.96$ & $0.43$\\
    vanilla \texttt{FLAIR} ($\rho=1$) & $92.81^{0.13}$ &$91.59^{0.10}$  &$1.15$ & $0.50$& &$92.46^{0.13}$ & $91.50^{0.09}$& $0.80$ & $0.35$\\
    &  \multicolumn{9}{c}{$p=10000$}  \\
&  \multicolumn{9}{c}{Scenario (a)}  \\
                    & \multicolumn{4}{c}{$n=500$}& &\multicolumn{4}{c}{$n=1000$}\\
& \multicolumn{2}{c}{Coverage}&\multicolumn{2}{c}{Length} & & \multicolumn{2}{c}{Coverage} &\multicolumn{2}{c}{Length} \\
    Method & $\Lambda \Lambda^\top$ & $B$  & $\Lambda \Lambda^\top$ & $B$ & & $\Lambda \Lambda^\top$ & $B$  & $\Lambda \Lambda^\top$ & $B$\\ 
    \texttt{LVHML}& $87.48^{0.46}$ &$87.43^{0.16}$  &$0.74$ & $0.35$& &$90.83^{0.17}$ & $87.30^{0.14}$& $0.52$ & $0.24$\\ %
    \texttt{FLAIR} &$95.02^{0.11}$ &$92.45^{0.14}$  &$0.84$ & $0.42$& &$94.99^{0.10}$ & $92.36^{0.12}$& $0.59$ & $0.29$\\
    vanilla \texttt{FLAIR} ($\rho=1$) & $91.25^{0.14}$ &$88.43^{0.16}$  &$0.73$ & $0.36$ & &$91.29^{0.12}$ & $88.34^{0.15}$& $0.51$ & $0.25$\\
&  \multicolumn{9}{c}{Scenario (b)}  \\
                    & \multicolumn{4}{c}{$n=500$}& &\multicolumn{4}{c}{$n=1000$}\\
& \multicolumn{2}{c}{Coverage}&\multicolumn{2}{c}{Length} & & \multicolumn{2}{c}{Coverage} &\multicolumn{2}{c}{Length} \\
    Method & $\Lambda \Lambda^\top$ & $B$  & $\Lambda \Lambda^\top$ & $B$ & & $\Lambda \Lambda^\top$ & $B$  & $\Lambda \Lambda^\top$ & $B$\\ 
    \texttt{LVHML}& $92.69^{0.09}$ &$90.33^{0.09}$  &$1.17$ & $0.51$& &$92.58^{0.11}$ & $90.70^{0.08}$& $0.77$ & $0.35$\\ %
    \texttt{FLAIR} &$96.55^{0.08}$ &$95.22^{0.07}$  &$1.23$ & $0.58$& &$96.24^{0.08}$ & $95.06^{0.08}$& $0.84$ & $0.41$\\
    vanilla \texttt{FLAIR} ($\rho=1$) & $93.31^{0.09}$ &$91.82^{0.10}$  &$1.06$ & $0.51$ & &$92.98^{0.11}$ & $91.70^{0.09}$& $ 0.73$ & $ 0.36$\\
    \end{tabular}}
    \label{tablelabel}
   \caption{Comparison of the methods in terms of uncertainty quantification in longitudinal simulation examples. 
    Average frequentist coverage for entries of a random $100\times 100$ submatrix of $\Lambda \Lambda^\top$ and $B$ and length of the corresponding intervals. We report mean and standard deviation over 50 replications. Coverage values have been multiplied by $10^2$. All the standard errors for the length of intervals were smaller than $0.01$ and omitted. \texttt{LVHML} denote \citet{lee24}'s method. For \texttt{FLAIR}, we report results with and without applying the correction factor $\rho$ to the posterior variance.}
    \label{tab:uq_longitudinal}
\end{table}

\subsection{Without Covariates scenarios}\label{subsec:no_covs}

We also consider examples without covariates and compare \texttt{FLAIR} to \citet{chen_jmle}'s method (\texttt{JMLE}, henceforth) using the R package \texttt{mirtjml}. \texttt{JMLE} considers model \eqref{eq:gllvm} with only the intercept and no additional covariates. Hence, \texttt{JMLE} can be considered as a special case of \texttt{LVHML} with $T=1$ and without covariates. 
As for \texttt{LVHML}, \texttt{JMLE} assumes the latent factors to be fixed constants and therefore, we make the same adjustment discussed in Section \ref{subsec:longitudinal_experiments} to estimate the latent covariance. 

We simulate data from the following model
\begin{eqnarray}
\mbox{pr}(y_{ij}=1|\eta_i) = 
h(\mu_{0j} + \lambda_{0j}^\top \eta_i),\quad 
\eta_i \sim N_k(0, I_k),\quad  (i=1,\ldots,n),\label{eq:gllvm}
\end{eqnarray}
where loadings and intercepts are generated as follows
\begin{equation*}
\lambda_{0jl} \sim 0.5 \delta_0  + 0.5 TN(0, \sigma^2, [-5, 5]), \quad \mu_{0j}\sim TN(0, \sigma^2, [-5, 5]) , \quad (l=1, \dots, k; j=1, \dots, p).
 \end{equation*}
We let the sample and outcome sizes be $(n, p) \in \{500, 1000\} \times \{1000, 10000\}$, and set $k=10$ and $\sigma^2 = 0.5$. 
 \texttt{JMLE} imposes a constraint on the Frobenius norm of the parameters and we find the results to be highly sensitive to this choice. To select this hyperparameter, we perform a $80\%/20\%$ train and test split, fit \texttt{JMLE} for each value in $\{1, 2, \dots, 10\}$ on the training set, and pick the value maximizing the area under the curve on the test data. Next, we refit \texttt{JMLE} on the entire data using the chosen hyperparameter. For each run, the optimal value of the hyperparameter was never on the boundary.
 
Table \ref{tab:accuracy_no_covariates} reports estimation accuracy and running times for both methods. We focus on the normalized Frobenius error for $\Lambda \Lambda^\top$, the root mean squared error for $\mu$, and the median and maximum absolute error for the entries of $\mu$. As in the longitudinal experiments, \texttt{FLAIR} obtains better estimation accuracy, which is particularly evident for $\Lambda \Lambda^\top$ and the maximum absolute error for entries of $\mu$. Moreover, \texttt{FLAIR} is remarkably faster, with a gain of a factor of $\approx 5$ when $p=10000$. This gain is for a single model fit that ignores the time taken to repeatedly run \texttt{JMLE} to tune the hyperparameters; \texttt{FLAIR} is run a single time on each data set without the need for such tuning. Table \ref{tab:uq_no_cov} provides further evidence on the validity of credible intervals of \texttt{FLAIR} with only minor under-coverage for $\mu$ when $p=10000$ and precise coverage for $\Lambda \Lambda^\top$.

\begin{table}[h]
\centering
{	\begin{tabular}{crrrrr}
		&  \multicolumn{5}{c}{$p=1000$, $n=500$}  \\
		Method & $\Lambda \Lambda^\top$ & $\mu$ & Median $ \mu_{j}'s$ & Max $\mu_j$'s & time (s)   \\
    \texttt{JMLE} &  $ 44.75^{0.38}$ & $ 14.55^{0.14}$& $9.20^{0.09}$&$64.23^{2.10}$  &  $17.27^{0.89}$  \\
		\texttt{FLAIR} &  $37.08^{0.14}$ & $ 13.51^{0.14}$& $8.80^{0.09}$&$52.27^{1.50}$ & $12.74^{0.58}$ \\
		&  \multicolumn{5}{c}{$p=1000$, $n=1000$}  \\
		Method & $\Lambda \Lambda^\top$ & $\mu$ & Median $ \mu_{j}'s$ & Max $\mu_j$'s & time (s)   \\
        \texttt{JMLE}  &  $ 31.15^{0.21}$ & $ 10.18^{0.09}$ & $6.46^{0.06}$  &$ 44.64^{1.08}$ & $ 37.84^{3.59}$   \\
		\texttt{FLAIR}  &  $ 25.68^{0.07}$ & $ 9.64^{0.09}$ & $6.24^{0.06}$ &$37.32^{0.63}$ & $ 20.35^{1.08}$   \\
        &  \multicolumn{5}{c}{$p=10000$, $n=500$}  \\
Method & $\Lambda \Lambda^\top$ & $\mu$ & Median $ \mu_{j}'s$ & Max $\mu_j$'s & time (s)   \\
        \texttt{JMLE}  & $ 39.78^{0.19}$  &$ 14.00^{0.13}$  & $9.07^{0.08}$  &$ 67.38^{0.95}$ & $ 190.44^{1.24}$   \\
		\texttt{FLAIR}  &  $ 38.31^{0.15}$ & $ 13.41^{0.13}$ & $8.77^{0.08}$ &$ 62.92^{1.22}$ & $ 30.77^{0.46}$   \\
        & \multicolumn{5}{c}{$p=10000$, $n=1000$} \\
Method & $\Lambda \Lambda^\top$ & $\mu$ & Median $ \mu_{j}'s$ & Max $\mu_j$'s & time (s)   \\
        \texttt{JMLE}  &  $ 27.08^{0.10}$ & $9.73^{0.08}$ & $6.32^{0.05}$  &$46.34^{0.71}$ & $ 286.96^{15.86}$   \\
		\texttt{FLAIR}  &  $26.18^{0.08}$ & $ 9.50^{0.08}$ & $6.20^{0.05}$ &$43.49^{0.71}$ & $ 45.37^{0.72}$   \\
	\end{tabular}}
	\caption{Comparison of the methods in terms of estimation accuracy in the simulation experiments without covariates. 
     Root normalized squared error for $\Lambda \Lambda^\top$ and $\mu$, median and maximum absolute error for entries of $\mu$, and running time. We report mean and standard deviation over 50 replications. Estimation errors have been multiplied by $10^2$. \texttt{JMLE} denote \citet{chen_jmle}'s method.}
	\label{tab:accuracy_no_covariates}
\end{table}

\begin{table}[h]
\centering
{	\begin{tabular}{crrcrr}
&  \multicolumn{5}{c}{$p=1000$}  \\
& \multicolumn{2}{r}{$n=500$}& &\multicolumn{2}{r}{$n=1000$} \\
		Method & $\Lambda \Lambda^\top$ & $\mu$ & & $\Lambda \Lambda^\top$ & $\mu$ \\ 
		\texttt{FLAIR} &$ 96.85^{0.09}$ &$ 94.52^{0.23}$ &  & $ 96.47^{0.08}$&$ 94.17^{0.25}$  \\ 
		vanilla \texttt{FLAIR} ($\rho=1$) & $ 92.65^{0.16}$ & $ 89.51^{0.37}$ & & $ 92.13^{0.12}$ & $ 88.99^{0.31}$\\ 
   &  \multicolumn{5}{c}{$p=10000$}  \\
& \multicolumn{2}{r}{$n=500$}& &\multicolumn{2}{r}{$n=1000$}\\
		Method & $\Lambda \Lambda^\top$ & $\mu$ & & $\Lambda \Lambda^\top$ & $\mu$ \\ 
\texttt{FLAIR} &$ 96.50^{0.07}$ &$ 92.01^{0.29}$ &  & $ 96.40^{0.08}$&$ 91.50^{0.31}$  \\ 
		vanilla \texttt{FLAIR} ($\rho=1$) &$92.77^{0.12}$ &$ 86.94^{0.35}$ &  &$ 92.28^{0.12}$ & $ 86.60^{0.36}$ \\  
	\end{tabular}}
	\caption{Frequentist coverage for \texttt{FLAIR} with and without applying the correction factor $\rho$ to the posterior variance in the simulation experiments without covariates. 
    Average frequentist coverage for entries of a random $100\times 100$ submatrix of $\Lambda \Lambda^\top$ and $\mu$ for equi-tailed 95\% credible intervals for  \texttt{FLAIR} in simulation studies for varying $n$ and $p$ with and without applying the correction factor $\rho$ to the posterior variance.
We report mean and standard error over 50 replications. All values have been multiplied by $10^2$. }
	\label{tab:uq_no_cov}
\end{table}

\subsection{Additional Details for the Numerical Experiments}
In the experiments in the main article, to select the hyperparameters of \texttt{GMF} controlling the $L_2$ penalty on $B$ and $\Lambda$, $\gamma_B$ and $\gamma_\Lambda$, we estimate the test sample predictive accuracy measured via the area under the curve for each combination of $(\gamma_B, \gamma_\Lambda) \in \{0, 0.5, 1, 5, 10\} \times \{0, 0.5, 1, 5, 10\}$ and refit the model on the entire dataset with the configuration maximizing the area under the curve. The number of latent factors was estimated using via the information criterion discussed in Section \ref{subsec:choice_k}, which always picked the correct value. 
For \texttt{FLAIR}, we always set $C_\Lambda$ and $C_B$ to 10. 
The code to implement the \texttt{FLAIR} methodology and replicate the experiments is available at \url{https://github.com/maurilorenzo/FLAIR/}. All experiments were run on a Laptop with 11th Gen Intel(R) Core(TM) i7-1165G7 @ 2.80GHz and 16GB RAM.

\subsection{Additional Details for the Application to Madagascar Arthropods Data}
For \texttt{GMF}, we test each configuration of $(\gamma_B, \gamma_\Lambda) \in \{0, 0.5, 1, 5, 10, 20, 50\} \times \{0, 0.5, 1, 5, 10, 20, 50\}$. For  \texttt{FLAIR}, we set $C_\Lambda$ and $C_B$ to 10. 
The initialization strategy of \texttt{FLAIR} described in Section \ref{subsec:svd_init} requires the application of a singular value decomposition to the original data matrix. 
To avoid using the hold out data in the procedure, we imputed each element of the held out data set. In particular, if the observation $y_{ij}$ was included in the hold out set, we replaced it by the product of the empirical means of the $i$-th row and the $j$-th column of $Y$.  We considered other imputation strategies but noticed negligible dependence of the final solution on this choice. Moreover, the \texttt{FLAIR} procedure was trivially modified to include in the calculation of the likelihood for the joint maximum {\em a posteriori} computation only observations in the training set.


\end{document}